\documentclass[sigconf, nonacm]{acmart}

\newcommand\vldbdoi{XX.XX/XXX.XX}

\newcommand\vldbavailabilityurl{https://anonymous.4open.science/r/optiaqp}
\newcommand\vldbpagestyle{plain}

\usepackage{times} 
\usepackage{bm} 
\usepackage{color} 
\usepackage{xcolor}
\usepackage{hyperref} 
\usepackage{graphicx} 
\usepackage{multirow}
\usepackage[tight]{subfigure}
\usepackage{verbatimbox}
\usepackage{tcolorbox}
\usepackage{wrapfig}
\usepackage{titlesec}
\usepackage{enumitem}
\usepackage{pifont}
\usepackage{tabularx}
\usepackage[ruled,linesnumbered,noend]{algorithm2e}
\usepackage{listings}
\usepackage{xspace}

\makeatletter \let\c@table\c@figure \makeatother

\DeclareMathOperator*{\argmin}{arg\,min}

\newenvironment{pkl}{%
\begin{itemize}%
\setlength\itemsep{-\parskip}%
\setlength\parsep{0in}%
}{%
\end{itemize}}

\newenvironment{enu}[1][label={\arabic*}.]{%
\begin{enumerate}[wide,#1]%
\vspace{-\topsep}%
\setlength\itemsep{-\parskip}%
\setlength\parsep{0in}%
}{%
\vspace{-\topsep}%
\end{enumerate}}

\setitemize{leftmargin=*}
\setenumerate{leftmargin=*}

\newcommand{\sysname}{OptiAQP\xspace}

\begin{document}

\title{Index-Assisted Stratified Sampling for Online Aggregation}

\author{Yunnan Yu}
\affiliation{%
  \institution{University at Buffalo}
}
\email{yunnanyu@buffalo.edu}

\author{Zhuoyue Zhao}
\affiliation{%
  \institution{University at Buffalo}
}
\email{zzhao35@buffalo.edu}




\newcommand{\revcut}[1]{}
\newcommand{\reva}[1]{{\color{black}{#1}}}
\newcommand{\revb}[1]{{\color{black}{#1}}}
\newcommand{\revc}[1]{{\color{black}{#1}}}
\newcommand{\zy}[1]{{\color{red}(ZY: #1)}}
\newcommand{\yn}[1]{{\color{blue}#1}}

\newcommand{\Paragraph}[1]{{\noindent \bf #1}}


\newcommand{\cut}[1]{}
\newcommand{\eat}[1]{}

\SetKwProg{Def}{def}{}{}
\SetKwFor{For}{for (}{)}{}
\SetKwFor{Foreach}{foreach (}{)}{}
\SetKwComment{Comment}{//}{}
\SetKw{Break}{break}
\SetKw{Continue}{continue}
\SetKwRepeat{Do}{do}{while}

\begin{abstract}

    Ad-hoc queries over frequently updated data in a flat schema are common in
	real-time data analysis applications and often require very low
	latency.  Online aggregation can achieve so by providing approximate
	aggregation answers with confidence bound guarantees. It relies on the
	ability to draw samples online in a linear time to sample size rather
	than database size, which can be supported by index-assisted
	Sampling-based Approximate Query Processing (S-AQP) systems.  However,
    \revc{the query latencies of approximate queries in these systems} can
    still suffer from excessive sampling cost required to achieve a desired
    confidence bound, due to increased sample
	size for data with high variance in value distribution and selectivity.
	Classic stratified sampling methods with Neyman allocation can minimize
	sample size in theory, but several challenges prevent it from being
	applicable in index-assisted S-AQP systems, including requiring apriori
	statistics, high optimization cost, and inaccurate sampling cost model
	based on sample size.  Towards that, we design index-assisted
	stratified sampling for online aggregation, which features a two-phase
	sampling framework. Samples drawn from first phase are used for both
	online aggregation and optimizing future sampling cost, while the
	second phase continues the online aggregation using the optimized
	strata. We prove optimal stratification and sample size allocation
	strategies for index-based sampling cost model, and design several
	greedy and dynamic programming based optimization methods to balance
	optimization cost and effectiveness in cost reduction. We evaluate our
	methods on several real-world and synthetic datasets and queries, and
    the results show ours \revc{consistently achieve good speedup
	and, in extreme cases, up to 3x speedup and
    98708x speedup, when compared to index-assisted uniform sampling and classic
    scan-based stratified sampling respectively}.

\end{abstract}

\maketitle

\pagestyle{\vldbpagestyle}

\ifdefempty{\vldbavailabilityurl}{}{
\vspace{.3cm}
\begingroup\small\noindent\raggedright\textbf{Artifact Availability:}\\
The source code, data, and/or other artifacts have been made available at \url{\vldbavailabilityurl}.
\endgroup
}


\vspace{-6 pt}
\section{Introduction}
\label{sec:intro}

\begin{figure}[t]
    \centering
    \includegraphics[width=.7\linewidth]{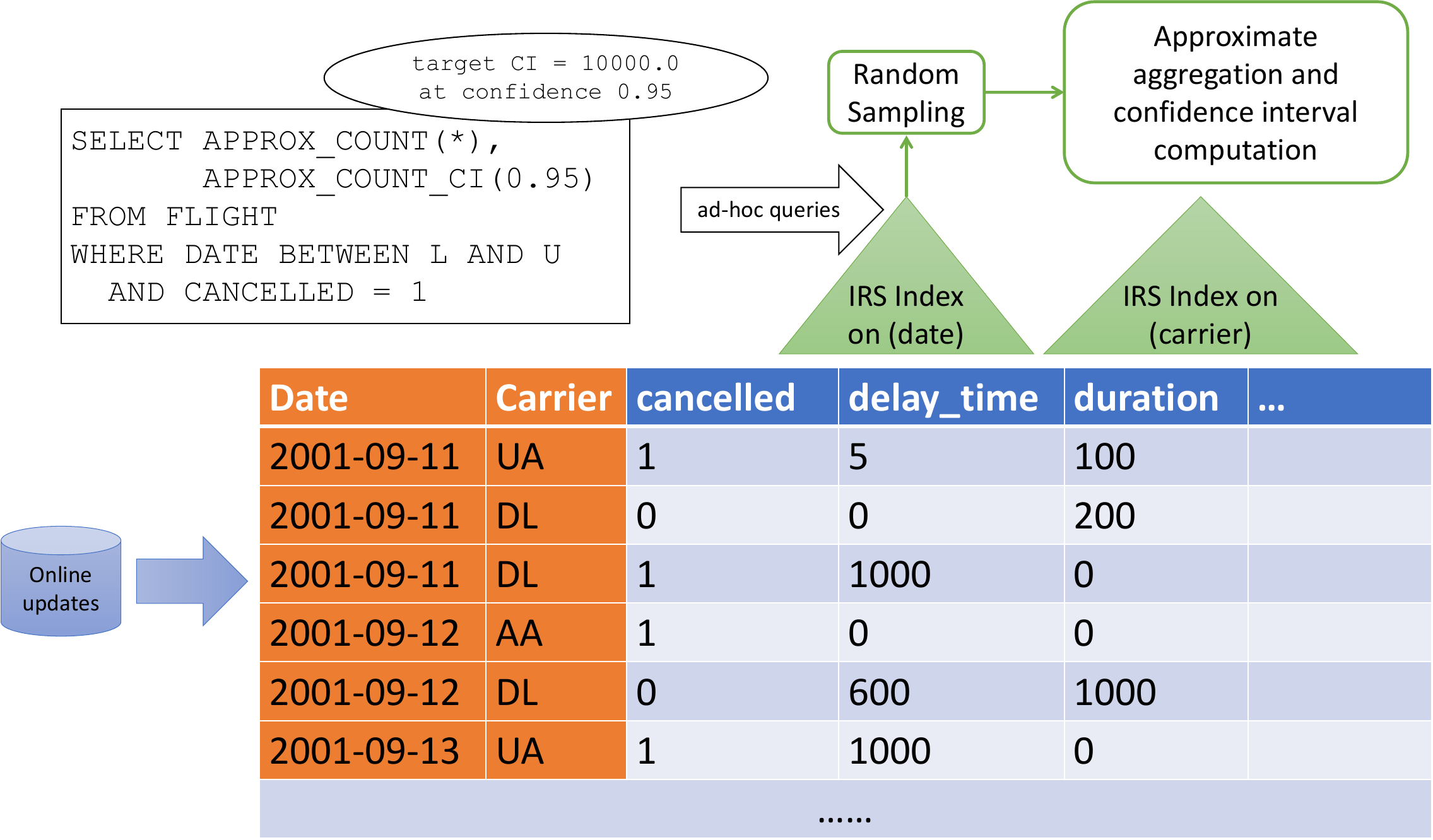}
    \vspace{-12 pt}
    \caption{Ad-hoc queries over in index-assisted S-AQP system}
	\vspace{-10pt}
    \label{fig:system_setup}
\end{figure}
\begin{figure}[t]
  \centering
\vspace{-7pt}
  \includegraphics[width=.7\linewidth]{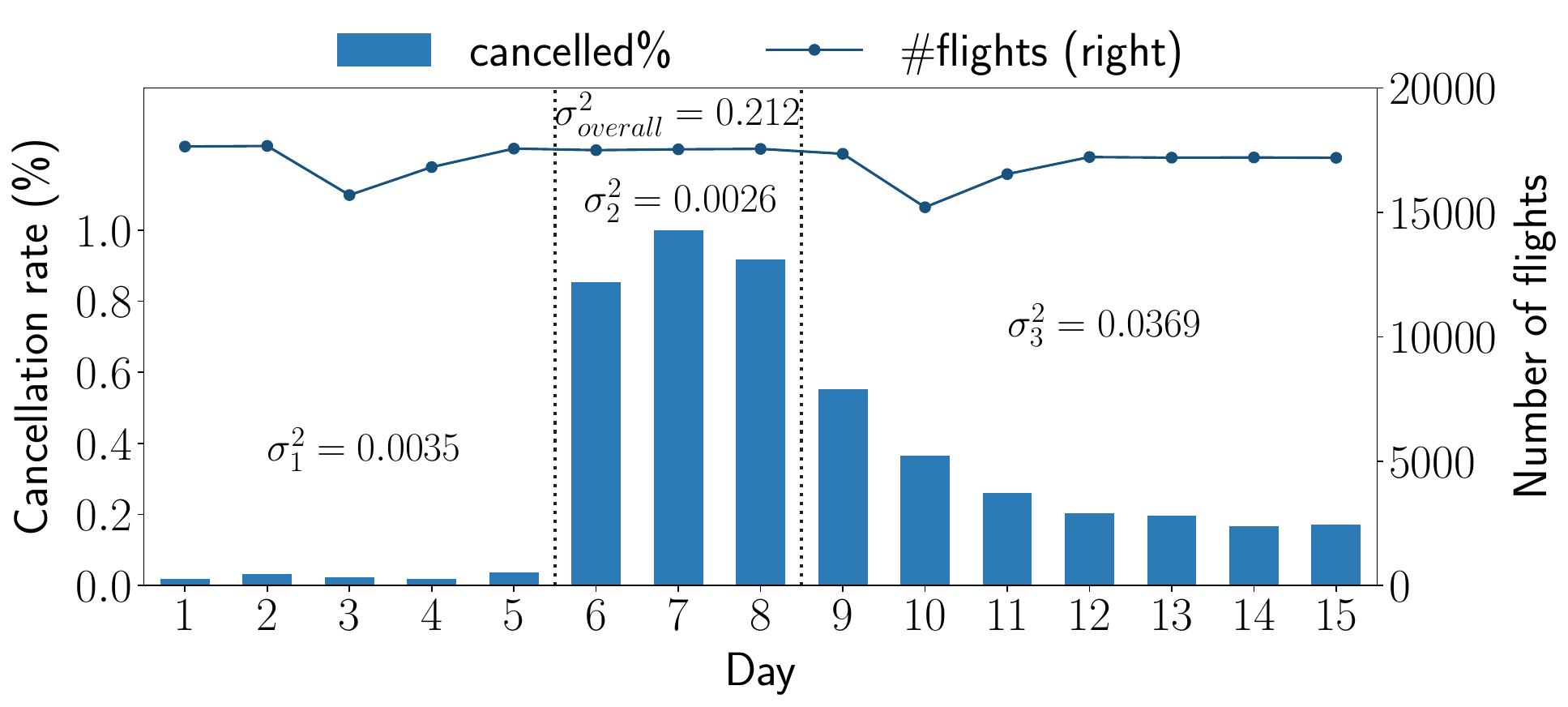}
  \vspace{-16 pt}
    \caption{\revc{The percentage of cancelled flights and the total flight
    during the date range. Overall estimator variance is higher than
    those of any of the three smaller date ranges.}
    \Description{The acceptance rates.}}
  \vspace{-5pt}
  \label{fig:flight_sample}
\end{figure}

Ad-hoc queries often require low latency in real-time business intelligence
and reporting applications~\cite{Ding:2016:SSA:2882903.2915249}, where users
often execute ad-hoc range queries over fresh data in a flat schema (such as
online ads log data, flight status monitoring, etc.).  Materialization and
online full scan are too expensive as they are large and frequently updated.
Online Aggregation~\cite{hellerstein97:_onlin} is a sampling-based approximate query processing
method for users to obtain approximate answers with low
response time, where samples must be drawn
online~\cite{10.1145/3517804.3526068} to ensure the correctness of results
under updates and independence of multiple queries.  To achieve so, state of
the art is to utilize index-assisted
sampling~\cite{hellerstein97:_onlin,10.14778/3611540.3611602,10.14778/3538598.3538606}
(Figure~\ref{fig:system_setup}), which can produce highly accurate results
within seconds due to the ability to draw samples with a cost (almost)
linear to sample size instead of database size.

However, the time to achieve a given confidence interval in these systems may
vary heavily. The root cause is \revc{these systems} often draw samples uniformly without
considering the variations of data distribution and query selectivity within
a dataset, which results in excessive amount of samples to be drawn
uniformly across an entire query range.
\revc{
As a concrete example (Figure~\ref{fig:system_setup}), consider an ad-hoc range
query over a real-world dataset: the US airline on-time performance
dataset~\cite{DVN/HG7NV7_2008}, on which one wants to find the number of
flights that were cancelled in some date range for analysis.  For most of the
date ranges, the cancellation rate in the US is consistently low, leading to
very accurate \texttt{COUNT} estimation with a small number of samples.
However, if we happen to query a range where part of the query range has
cancellation spikes due to special events (e.g., a snow storm)
(Figure~\ref{fig:flight_sample}), the estimator variance will be extremely
high, leading to high latency to achieve a desired confidence bound.  On the
other hand, if we query the three smaller ranges as shown in
Figure~\ref{fig:flight_sample} and combine their estimators instead, all of
them have smaller estimator variances, which could lead to much lower latency.}



This is known as stratified sampling, which breaks the entire query range into
smaller partitions and independently sample from each partition and can often
reduce the total sample needed and thus reduce query time. However, all known
stratified sampling in AQP systems require scanning the database table for
preprocessing online or offline, which either cause long preprocessing delays
during query time or cannot handle data updates efficiently.  As examples,
VerdictDB, a state-of-the-art AQP middleware for existing
DBMS~\cite{10.1145/3183713.3196905}, uses a form of stratified sampling where
stratification is done over distinct values of range predicate columns, and
sample size is allocated evenly to different partitions.
QuickR~\cite{Kandula:2016:QLA:2882903.2882940} performs online scan-based
sampling using a slightly improved version called distinct sampling, where
sample size allocation across distinct values where a minimum sample size is
guaranteed for each distinct value.  In~\cite{10.1145/1242524.1242526}, Surajit
et el. optimizes offline samplings for query workload, which utilizes a
well-known technique, Neyman allocation, to optimally allocate sample sizes to
different partitions in order to minimize total sample size needed. 
none of these works it targets offline sampling scenario where the
The crux of the problem is having to know distinct keys and/or sum statistics to
create strata and/or perform sample size allocation, which currently cannot be
easily derived in an online fashion. In addition, traditional methods only
consider sample size as a simple cost model for sampling-based approximate
query evaluation, but ignore the fact that sampling in different strata can
often have very different per-sample cost.

Motivated by that, we propose a two-phase index-assisted approximate query
evaluation framework based on stratification, \sysname, which can use a small
initial sample to obtain a number of candidate partition boundaries and
estimated statistics to derive an optimized stratification strategy for the
subsequent approximate query cost until a desired confidence bound is achieved.
Combining the existing concurrency-safe index-assisted
sampler~\cite{10.14778/3538598.3538606} that can draw samples from
a latest snapshot at any time under snapshot isolation, our method will
significantly reduce latencies of ad-hoc approximate aggregation
queries over frequently updated data with no apriori statistics.
The main challenges in our proposed method are threefold: (1) how to effectively
utilize sampling indexes to collect initial samples for stratification; (2) how
to balance the first phase optimization cost and the sample cost reduction it
results in the second phase; (3) how to accurately model the sampling cost.
To summarize, We claim the following contribution:
\vspace{-2 pt}
\begin{pkl}
    \item We describe how to perform index-assisted stratified sampling and
        analyze the cost model of approximate query evaluation in
        index-assisted S-AQP system with stratification.
    \item We prove the optimal stratification and sample size allocation
        strategies under the cost model given full statistics.
    \item We design \sysname, a two-phase index-assisted approximate query
        evaluation framework and introduce 4 stratification optimization
        methods for \sysname with different trade-offs between optimization
        overhead and cost reduction.
    \item We implement \sysname in PostgreSQL based on a state-of-the-art
        index-assisted S-AQP system~\cite{10.14778/3611540.3611602,10.14778/3538598.3538606}.
        We validate the query time reduction of
        \sysname over existing approaches compared to several baselines over
        several real-world datasets and a synthetic benchmark based on the
        skewed TPC-H.
    \item We discuss possible future extensions for handling more query
        types including joins, spatial ranges, and group-by queries,
        and the additional design space.
\end{pkl}
\vspace{-2 pt}
The rest of the paper is organized as follows. In Section~\ref{sec:background},
we provide necessary background on independent range sampling and the basic
S-AQP algorithms in index-assisted S-AQP systems. In
Section~\ref{sec:analysis}, we provide analysis of the cost model of
index-assisted stratified sampling and the optimal stratification and sample
size allocation strategy.  Based on the analysis, we provide the design and
implementation details of \sysname in Section~\ref{sec:algorithm}. We provide
extensive evaluation in Section~\ref{sec:exp}. We survey related works and
provide a discussion of possible extensions in Section~\ref{sec:related} and
conclude the paper in Section~\ref{sec:conclusion}.

\section{Problem Formulation and Background}
\label{sec:background}

\begin{table}[t]
    \caption{List of notations in this paper}
	\vspace{-10pt}
    \label{tab:notations}
    {\footnotesize
    \begin{tabular}{r|l}
	\hline
        $T$ & a relational table \\
        $\mathcal{Q}, \tilde{\mathcal{Q}}$ & exact query and approximate query\\
        $A, \tilde{A}$ & exact answer and its unbiased estimator\\ 
        $\varepsilon$ & half-width of confidence interval \\
        $1 - \delta$ & confidence level \\
        $\gamma_{AGG(e)}(\cdot)$ & aggregation operator~\cite{DBLP:books/mg/SKS20} that computes $AGG(\pi_e(\cdot))$\\
        $\Gamma_n(\cdot)$ & i.i.d. sampling with replacement of size $n$ over ($\cdot$)\\
        $\mathcal{P}, \mathcal{P}_r, \mathcal{P}_f$ & overall selection, range and additional filter predicates \\
        $x$ & range predicate column (set) \\
        $L, U$ & lower/upper bounds in range column \\
	$S$ & bag of sampled tuples \\
	$c, c_0$ & total cost and preprocessing factor\\
	$D$ & query range\\
	$k, \Delta k$ & number of strata and number of children of $D'$\\
	$\tau$ & stopping threshold for \texttt{Greedy}\\
	$\vec{m}$ & cumulative moment statistics\\
	$\vec{B}, \vec{C}$ & stratum boundaries and collect candidate boundaries\\
	$h, H$ & height of the LCA and the tree root\\
	$d$ & partition granularity\\
	\hline
    \end{tabular}
    }
\end{table}


In this section, we first provide an overview of problem formulation, and then
introduce necessary background on online aggregation style approximate
evaluation and sampling indexes. Table~\ref{tab:notations} lists the notations
in this paper. 

\noindent\textbf{Problem Formulation.} In this work, we study the problem of
reducing online aggregation style approximate query evaluation cost in an
index-assisted S-AQP system.  Without loss of generality, let $T$ be a
relational table and we consider a single-table approximate aggregation query
in the form of

\vspace{-10 pt}
\begin{equation}
    \label{eqn:original_query}
\mathcal{Q} = \gamma_{\texttt{SUM}(e)} \sigma_{\mathcal{P}}
T
\vspace{-2 pt}
\end{equation}

\noindent where $\mathcal{P}$ is a conjunctive predicate. It can be written as
$\mathcal{P}_r \wedge \mathcal{P}_f$. Here, $\mathcal{P}_r$ is a range
predicate $x \in [L, U)$ and there is a sampling index~\cite{hu14:_indep,
afshani_et_al:LIPIcs:2017:7859,olken93:_random,10.14778/3538598.3538606,10.1145/3626744}
over $x$ which can draw independent samples from $\sigma_{x \in [L, U)} T$
without scanning table $T$.  For the rest of the paper, we denote the
\textit{Independent Range Sampling (IRS)}
operation~\cite{10.1145/3517804.3526068} of sample size $n$ with replacement as
$\mathbf{\Gamma_n \sigma_{\mathcal{P}_r} T}$ -- i.e., $\Gamma_n$ is a sampling
operator that draws an i.i.d. random samples with some distribution (e.g., uniform/weighted) over
$\sigma_{\mathcal{P}_r} T$.  In addition, the sampling index also
appends to the output schema of $\sigma_{\mathcal{P}_r} T$ an additional
\textbf{probability column} $p$ to indicate each sample's probability.
The remaining part of the selection predicate
$\mathcal{P}_f$ is an additional filter on top of the range filter. Let $A$ be
the (unknown) answer of the exact query $\mathcal{Q}$. An approximate
evaluation of $\mathcal{Q}$ yields an unbiased estimation $\tilde{A}$ such that
$E[\tilde{A}] = A$, and a confidence bound $(\varepsilon, \delta)$ such that
$Pr\{\tilde{A} \in [A - \varepsilon, A + \varepsilon]\} \geq 1 - \delta$.
Typically, user specifies the desirable (half-width) confidence interval
$\varepsilon$ and confidence level $1 - \delta$, and the goal is to achieve the
desired confidence bound as fast as possible.

\noindent \textbf{Basic S-AQP Algorithms.}
To obtain an approximate answer of query $\mathcal{Q}$ is to evaluate its
approximate counterpart:
$$\tilde{\mathcal{Q}}^{(n)} =
\gamma_{\tilde{A},
\varepsilon} \Gamma_n\sigma_{\mathcal{P}_{r}} T$$
incrementally with increasing sample size $n$ until the derived confidence
interval $\varepsilon$ is smaller than user specification $\varepsilon_0$.  In
more details, the evaluation of $\tilde{\mathcal{Q}}^{(n)}$ works in rounds by invoking the
sampling index for obtaining a delta sample bag $\Delta \mathcal{S} =
\Gamma_{\Delta n}\sigma_{\mathcal{P}_{r}} T$. 
To derive the current estimator
and confidence interval after a certain round, we first denote the union of all
delta sample bags as $\mathcal{S} = \bigcup \Delta \mathcal{S}$, and the total
sample size so far as $n = \sum \Delta n$. Note that the estimator $\tilde{A}$ and the
confidence interval $\varepsilon$ both depend on $\mathcal{P}_f$ as it is not
evaluated by the sampling index, but we omit them when they are
unambiguous in the context to simplify the notations. Their
formulae~\cite{haas97:_large} are given below based on the Horvitz-Thompson
estimator~\cite{horvitz52} and the Lindeberge-L\'evy Central Limit
Theorem~\cite[Thm.~27.1]{alma991026726359704801}\footnote{Note that, the
computed half width of CI $\varepsilon$ is an estimation due to the fact that the true
estimator variance $\sigma^2$ is unknown and can only be estimated from sample
variance $\tilde{\sigma^2}$.}


\vspace{-5 pt}
\begin{equation}
    \label{eqn:individual_unbiased_estimator}
    \tilde{A}(t) = \begin{cases}
        e(t)/p(t) & \text{if } \mathcal{P}_f(t) \\
        0 & \text{o.w.}
    \end{cases}
\end{equation}
\begin{equation}
    \label{eqn:unbiased_estimator}
    \tilde{A} = \tilde{A}(\mathcal{S}) = \sum_{t \in \mathcal{S}}
    \tilde{A}(t) / n
\end{equation}
\begin{equation}
\label{eqn:ci_bound_classic}
\varepsilon = \frac{Z_\delta \sigma}{\sqrt{n}}, Z_\delta = \sqrt{2} \texttt{erf}^{-1}(1 - \delta)
\end{equation}
\begin{equation}
\label{eqn:stddev_estimation}
\sigma \approx \tilde{\sigma} = 
    \sqrt{\frac{\sum_{t\in\mathcal{S}}{(\tilde{A}(t) - \tilde{A}})^2}{n-1}}
\end{equation}

\noindent \textbf{Sampling Indexes.}
All existing sampling indexes~\cite{hu14:_indep,
afshani_et_al:LIPIcs:2017:7859,olken93:_random,10.14778/3538598.3538606,10.1145/3626744}
are tree indexes. In general, the cost of these sampling
indexes for drawing $k$ independent random records from a table of size $N$
can modeled as $c^{pre}(N) + k\cdot c^{samp}(N)$, where $c^{pre}(N)$ and
$c^{samp}(N)$ are two sample-index dependent factors bounded by $O(\log N)$.

\begin{figure}[t]
  \centering
  \includegraphics[width=0.6\linewidth]{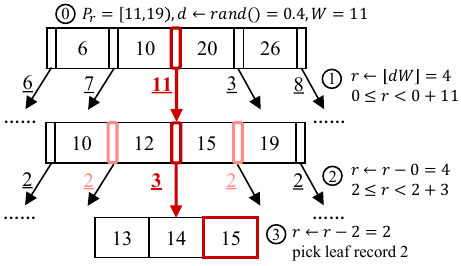}
  \vspace{-10pt}
  \caption{Sampling from an aggregate B-tree. Underscored numbers are aggregate weights for child pointers.}
  \vspace{-7pt}
  \label{fig:sampling_index}
\end{figure}

In this work, we consider the classic aggregate B-tree first described by
Olken~\cite{olken93:_random} because it is the only one with open-source DBMS
integration to the best of our knowledge~\cite{10.14778/3538598.3538606}.
Nevertheless, most of analysis below can adapt to other sampling indexes with
plugging in their specific cost models. In aggregate B-tree, each leaf record
is assigned an unnormalized weight such that the sampling probability for each
leaf record is proportional to its weight -- for uniform sampling, the weight
of each leaf record is simply $1$.  It also has each of its
internal node child pointer annotated with the aggregate weight of the subtree
it points to.

The original Olken's approach involves performing random selection of child
pointers at each level proportional to the aggregate weights, incurring an
$O(\log N)$ per-sample cost. \cut{This, however, requires generating $\log N$
random numbers, which could be expensive in real systems.} Instead,
\cite{10.14778/3538598.3538606} describes a slightly different procedure
(Figure~\ref{fig:sampling_index}), which has the same complexity but only
requires generating one random number per sample. Sampling $k$ records from a
range starts with a preprocessing stage to find the left-most and right-most
paths leading to the first and last records in the range, and compute the total
weight $W$ between these two paths. This takes $c_{pre}(N) = O(\log N)$ time.
Then, for each sample to draw, it draws a random number $d$ in $[0, 1)$, and
performs a weight-guided descent in the tree to identify the $\lfloor dW
\rfloor^{\textrm{th}}$ leaf record in the range: It first sets the residual
weight as $r = \lfloor dW \rfloor$. Starting from root, it computes the prefix
sum of the aggregate weights associated with the child pointers and find the
last one such that $r$ is no larger than the prefix sum of aggregate weights,
which is guaranteed to contain the $\lfloor DW \rfloor^{\textrm{th}}$ leaf
record in the query range.  Then it descends down into that subtree with the
residual weight $r$ subtracted by the prefix sum of aggregate weights of the
left sibling of the chosen child pointer.  This results in a cost of
$c^{sample}(N) = O(\log N)$ per sample, but as we will show later, this
modified procedure can help with additional cost reduction for stratified
sampling. Note that we assume each node visited incurs a non-negligible
constant cost, including I/O cost, locking overhead, and the linear cost for
searching for the child node pointing to the subtree that contains the $\lfloor dW
\rfloor^{\textrm{th}}$ leaf record.



%
%
%
%
%
%

\section{Optimal Stratification Strategies}
\label{sec:analysis}

This section provides a theoretical foundation of \sysname -- how we can derive
the optimal stratification strategies assuming we can estimate statistics of
the data. Towards that, we first describe how index-assisted stratified
sampling works and its cost model. Then we show the optimal sample size
allocation through a modified Neyman allocation assuming strata are fixed.
Finally, we prove how we can optimally determine the strata given full statistics.

\subsection{Online aggregation with index-assisted stratified sampling and its cost model}

Recall that the range predicate in the query $\mathcal{P}_r$ is in the form of $x
\in [L, U)$ and we denote $D = [L, U)$. Then online aggregation
with stratified sampling works as follows:
\begin{enu} 
    \item We \emph{logically} partition $D$ into two or more \emph{non-empty}
        disjoint ranges $D = D_1 \uplus D_2 \ldots \uplus D_k$. The disjoint
        ranges $D_1, \ldots D_k$ are also known as \textbf{strata}.
    \item The original query can be rewritten as $\mathcal{Q} =
        \gamma_{SUM}\bigcup_{i = 1}^k \mathcal{Q}_i$, where each
        subquery $\mathcal{Q}_i $ is the partial aggregation over the subrange
        $D_i$: $\mathcal{Q}_i = \gamma_{\texttt{SUM}(e)}\sigma_{x \in
        D_i \wedge \mathcal{P}_f} T$.
    \item We perform $k$ independent approximate evaluations of the subqueries
        $\tilde{\mathcal{Q}}_1^{(n_1)}, \ldots \tilde{\mathcal{Q}}_k^{(n_k)}$
        with certain sample sizes $n_1, \ldots n_k$, and aggregate all
        estimators and confidence bounds to derive the final estimation using
        Equations~\ref{eqn:combined_estimator} and~\ref{eqn:combined_ci} below.
        In other words, $\tilde{\mathcal{Q}}' = \gamma_{\tilde{A}',
        \varepsilon'}\bigcup_{i = 1}^k \tilde{\mathcal{Q}}_i^{(n_i)}$.
\end{enu}
\vspace{-10 pt}
\begin{equation}
    \label{eqn:combined_estimator}
    \tilde{A}' = \sum_{i = 1}^k \tilde{A}_i
\end{equation}
\vspace{-2 pt}
\begin{equation}
    \label{eqn:combined_ci}
    \varepsilon' = Z_\delta \sqrt{\sum_{i=1}^k\frac{\sigma_i^2}{n_i}} = \sqrt{\sum_{i=1}^k \varepsilon_i^2}
 \vspace{-2 pt}
\end{equation}

\noindent where $\tilde{A}_i$ and $\varepsilon_i$ are the unbiased estimator and half-width
of confidence interval produced by $\tilde{\mathcal{Q}}_i$.
We also denote the estimated estimator variance of $\mathcal{Q}_i$ as $\sigma_i^2$.
Then $\tilde{A'}$ remains an unbiased estimation of $A$ due to the linearity
of expectation, and $\varepsilon'$ remains to the estimated half-width of
confidence interval of $\tilde{A'}$ such that $Pr\{\tilde{A}' \in [A -
\varepsilon, A + \varepsilon]\} \geq 1 - \delta$, by modeling the distribution
of $\tilde{A}'$ as the sum of $k$ independent normal variables $\tilde{A}'_1,
\ldots, \tilde{A}'_k$. We omit the proofs and refer interested readers to
our supplemental materials.

To model the cost of the query with sample size $n_1, \ldots, n_k$, an obvious
bound is linear to the number of strata and the total number of samples taken
$c = c^{pre}(N) \cdot k + c^{sample}(N)\sum_{i=1}^k n_i$. However, we argue
this is often a loose overestimation in practice due to a simple implementation
trick of the sampling procedure below: When
sampling from a subrange $D_i = [L_i, U_i)$, we can first find the left-most
and right-most paths leading to the first and last record in $D_i$ with two
regular tree search descents from root.  Let the lowest common ancestor of the
two paths be $LCA$ and let the height of the $LCA$ and the tree root be $h$ and
$H$. Then, we know any leaf record in $D_i$ must have a unique ancestor node
for levels above $h$.  Thus, in the sampling algorithm, we can simply skip
visiting any node above level $h$ because there is no random choice to make,
which results in a lower upper bound of per-sample cost of $\frac{h}{H} \cdot
c^{sample}(N)$. Taking Figure~\ref{fig:sampling_index} as an example,
the LCA for the left-most and right-most paths for query range $[11, 19)$ is
the middle node at height 2. To draw samples from this range,
instead of starting the weight-guided descent from the root at height 3, which
can make no other choices but to descend down into the middle node, we can always
start from the middle node. This reduces the per-sample cost from $3$ to $2$.

There are two reasons why we have to consider the cost: (1) As strata in
stratified sampling are often much smaller than the entire key range of the
range column $x$, $h$ is often much lower than the entire tree height, this can
result in significantly lowered cost. (2) As different strata may have LCA at
different height, the per-sample cost in different stratum can also vary
significantly.
For example, consider the flight cancellation count query in
Figure~\ref{fig:flight_sample}. Suppose each day has $200k$ scheduled flights
and we have $20$ years of data. The total tree height of a sampling index with
fanout $50$ is $6$ -- thus per-sample cost is bounded by $6$. Suppose we are
querying the 20-day cancellation count since 9/11/2021. As there are
significantly more cancellation on 9/11 than others, we might want to create
two strata 9/11, and 9/12 - 9/30. With the same tree fanout and number of
scheduled flights, the first stratum will have a per-sample cost of $\log_{50}
100K \approx 3$, while the second stratum will have a per-sample cost of
$\log_{50} 1.9M \approx 4$. Both are much smaller than sampling in the entire
tree, and also differ by $1/3$ from each other. Hence, we must take height into
account when modeling the cost.

Thus, with an appropriately re-scaled constant $c_0 = c^{sample}(N)/H$
(\textbf{preprocessing factor}), we model the total cost of index-assisted stratified sampling for
$k$ strata with sample sizes $n_1, \ldots n_k$ as:
\vspace{-5pt}
\begin{equation}
    \label{eqn:cost_model}
    c = c_0 k + \sum_{k=1}^k n_i h_i
\end{equation}
\vspace{-15pt}

\subsection{Optimal sample size allocation via modified Neyman allocation
given fixed strata}
\label{sec:neyman_allocation}

\begin{lemma}
    \label{lem:neyman_allocation}
    (\textbf{Neyman allocation}) Given a set of strata $D = D_1 \uplus \ldots \uplus
    D_k$ and the estimator variances $\sigma_1^2, \ldots \sigma_k^2$ within
    each stratum, the minimum total sample size for approximate query
    $\tilde{\mathcal{Q}}'$ to reach a $(\varepsilon, \delta)$ confidence bound is $n' =
    \frac{Z_\delta^2}{\varepsilon^2}(\sum_{i=1}^k \sigma_i)^2$, where the
    sample size for each stratum is $n_i = \frac{Z_\delta^2}{\varepsilon^2}
    (\sum_{i=1}^k \sigma_i)\sigma_i$.
\end{lemma}

Neyman allocation~\cite{10.1145/1242524.1242526} (Lemma~\ref{lem:neyman_allocation})
provides the optimal
sample size allocation with respect to minimizing total sample size given a set
of strata. However, it fails to consider the different sampling cost in
each stratum, which can vary a lot as we discussed. Hence, we minimize the total cost rather than the total
sample size below. Note that here all strata are fixed so $k$ and $\sigma_i^2$
are constants.
\vspace{-5pt}
\begin{align*}
    \textrm{minimize } & c_0k + \sum_{i = 1}^k n_i\cdot h_i \\
    \textrm{subject to } & \varepsilon = Z_\delta \sqrt{\sum_{i=1}^k\frac{\sigma_i^2}{n_i}}
\vspace{-5 pt}
\end{align*}

Solving for this using the standard Lagrange multiplier method results in
Lemma~\ref{lem:modified_neyman_allocation}, which we call modified Neyman
allocation.  We omit the proof for brevity, which is available in the
supplemental materials. Note that, technically, $n_1, \ldots, n_k$ must be
integers, but the difference between the optimal and the rounded up $n'$ is
very small (bounded by $k$). Even if we round up all $n_i$'s, it only incurs
negligible increase of ultimate cost.

\begin{lemma}
    \label{lem:modified_neyman_allocation}
    (\textbf{Modified Neyman allocation}) Given a set of strata $D = D_1 \uplus \ldots \uplus
    D_k$ and the estimator variances $\sigma_1^2, \ldots \sigma_k^2$ within
    each stratum, the minimum index-assisted sampling cost for approximate query
    $\tilde{\mathcal{Q}}'$ to reach a $(\varepsilon, \delta)$ confidence bound is
    $c = c_0k + 
	\frac{Z_\delta^2}{\varepsilon^2}(\sum_{i=1}^k\sigma_i\sqrt{h_i})^2$, where the
    sample size for each stratum is $n_i = \frac{Z_\delta^2}{\varepsilon^2}
    (\sum_{i=1}^k \sigma_i\sqrt{h_i})\frac{\sigma_i}{\sqrt{h_i}}$.
\end{lemma}

\subsection{Optimal stratification given full statistics}

Finally, we consider the optimal stratification, assuming we will apply the
modified Neyman allocation and we are able to know the estimator variances and
possible candidate partition boundaries. Plugging in the $n_i$'s from
Lemma~\ref{lem:modified_neyman_allocation} into the cost function
(Equation~\ref{eqn:cost_model}), we essentially need to minimize the function
below (Equation~\ref{eqn:cost_optimization_2}). Note that the free variables in
this function are the number of strata $k$, and the stratum boundaries, which
impacts $h_i$ and $\sigma_i$.

\vspace{-15 pt}
\begin{align}
    \label{eqn:cost_optimization_2}
    c_{opt} = c_0k + (\sum_{i = 1}^k \sigma_i\sqrt{h_i})^2
\end{align}


The first term $c_0k$, which corresponds to preprocessing cost, obviously increases when there's a finer stratification.
If we can show that the minimal of the second squared summation term, which we denote
as $c_{opt}' = (\sum_{i = 1}^k \sigma_i\sqrt{h_i})^2$ and represents the squared cost
of actual drawing samples, over all possible
$k$-stratification only
decreases as $k$ increases, then $c_{opt}$ is a V-shaped function with a global
minimum. Intuitively, this is true -- the finer stratification is, the smaller
the second summation term of $c_{opt}$ will be, for having lower intra-stratum
estimator variance and lower tree height of the LCA for each stratum. To
formalize that, we argue that Theorem~\ref{thm:partition_is_better_general} is
true. We show a proof sketch below.

\begin{theorem}
    \label{thm:partition_is_better_general}
    \vspace{-4pt}
    Let $n$ be all the possible stratification boundary keys (i.e., distinct keys in the query range).
    For any $k$ disjoint strata of the range $D = \uplus_{i = 1}^k D_i$,
    where $D_1, \ldots, D_k$ are non-empty, 
    $c_{opt}' = (\sum_{i = 1}^k \sigma_i\sqrt{h_i})^2$ for index-assisted
    stratified sampling in $\tilde{\mathcal{Q}}'$ is no larger than that with
    $k - 1$ disjoint strata $\forall 1 \leq i' < i'' \leq n, D = \uplus_{1
    \leq i \leq n \wedge i \neq i' \wedge i \neq i''} D_i \uplus D_{i', i''}$,
    where $D_{i', i''} = D_{i'} \uplus D_{i''}$ is treated as a single
    stratum.
    \vspace{-4pt}
\end{theorem}
\begin{proof}
    Without loss of generality, we only need to prove so for $i' = 1$ and
    $i'' = 2$. It suffices to show that $\sigma_l \sqrt{h_1} + \sigma_2
    \sqrt{h_2} \leq \sigma_{1,2}\sqrt{h_{1,2}}$, where $\sigma_{1,2}^2$ and
    $h_{1,2}$ represent the estimator variance and LCA in sampling tree for
    sampling in stratum $D_{1, 2}$.

\begin{figure}[t]
  \centering
  \includegraphics[width=0.7\linewidth]{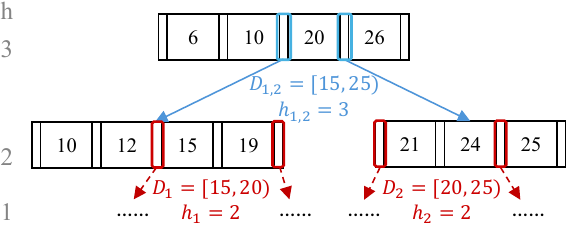}
  \vspace{-10 pt}
  \caption{Union of two strata has a higher LCA height than either stratum,
    resulting in higher per-sample cost}
  \label{fig:tree_height_is_larger_for_coarser_stratum}
  \vspace{-2pt}
\end{figure}
\begin{figure}[t]
    \centering
    \vspace{-10pt}
    \includegraphics[width=0.85\linewidth]{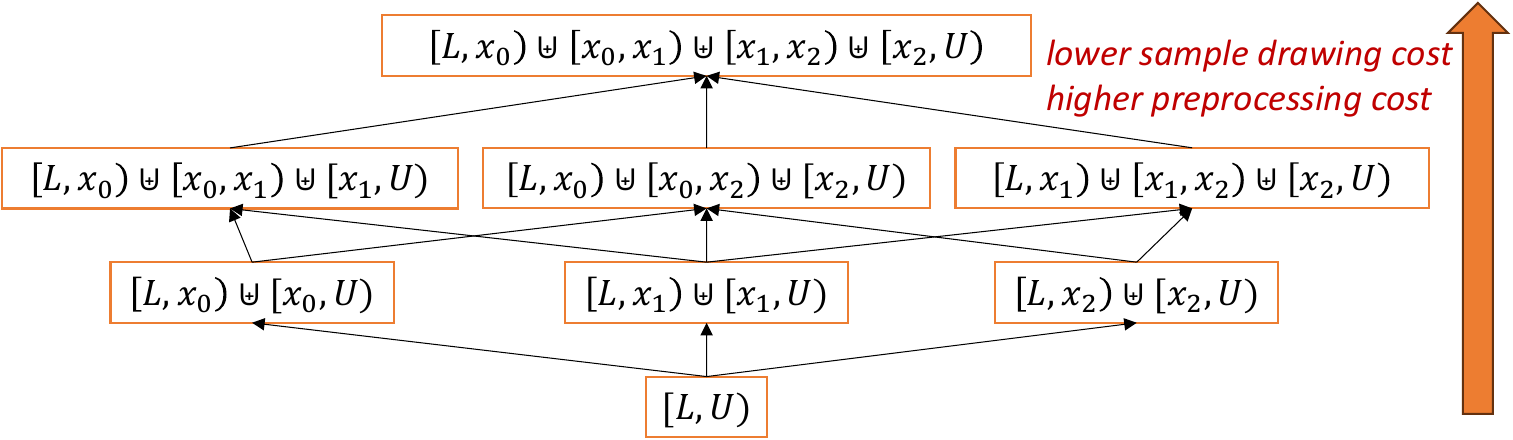}
	\vspace{-12pt}
    \caption{Possible stratification forms a lattice with
    decreasing sample drawing cost and higher preprocessing cost}
    \vspace{-5pt}
    \label{fig:lattice}
\end{figure}

    We first note that, the LCA of stratum $D_{1, 2}$ must be at the same
    height or higher height of the LCAs of $D_1$ and $D_2$, because $D_{1, 2}$
    covers a larger range (see
    Figure~\ref{fig:tree_height_is_larger_for_coarser_stratum}). Thus, we have
    $h_{1, 2} \geq h_1$ and $h_{1, 2} \geq h_2$.  Thus, it suffices to prove
    that $\sigma_1 + \sigma_2 \leq \sigma_{1, 2}$.

    Then, we can analytically compute the estimator variance below for
    stratum $i \in \{1, 2\}$, where $\tilde{A_i}(t)$ is the individual
    estimator variance defined in
    Equation~\ref{eqn:individual_unbiased_estimator}. Here we denote the
    fraction of tuples in stratum i out of both strata as $\alpha_i$, the
    total number of tuples in both strata as $N$ and assume uniform distribution
    as the sampling distribution:
    \vspace{-5 pt}
    \begin{align*}
        \sigma_i^2 &= E[\tilde{A_i}(t)^2] - E^2[\tilde{A_i}(t)] \\
        &= (\alpha_i N)^2 [
            \sum_{t \in D_i} \frac{(e(t)\mathcal{P}_f(t))^2}{\alpha_i N} - 
            (\sum_{t \in D_i} \frac{e(t)\mathcal{P}_f(t)}{\alpha_i N})^2]
    \end{align*}
    \vspace{-5pt}

    Then it is easy to prove $\sigma_1 + \sigma_2 \leq \sigma_{1, 2}$ by applying
    Jensen's inequality twice (omitted for space). In addition,
    if the sampling distribution is not uniform, we can apply the generalized
    Jensen's inequality instead, which provides the same conclusion.
\end{proof}

Interesting, this gives rise to a lattice structure (Figure~\ref{fig:lattice}) where finer stratification
that breaks a coarser stratum into two finer ones always leads to lowered
sampling drawing cost (i.e., second term $c_{opt}'$). Thus, one can either
holistically optimize $c_{opt}$ with increasing $k$ from bottom up, which is
the basis of the dynamic programming algorithm \texttt{CostOpt} in the next section,
or start from a single query range and gradually break it down
into smaller strata, which is the greedy algorithm in the next section.
The main difference is the trade-off between optimization
overhead and cost reduction. Top-down strategy often has
lower optimization cost but possibly lower reduction in $c_{opt}'$ while
bottom-up strategy has higher optimization cost but more reduction in
$c_{opt}'$.

\section{\sysname design and implementation}
\label{sec:algorithm}

In this section, we describe the design and implementation of \sysname based on
the previous theoretical analysis. As we need to know the statistics and
candidate partition keys for stratification, we first design a two-phase
index-assisted approximate query evaluation framework, where the first phase
draws initial samples using the available sampling index to compute an optimized
stratification for the second phase of index-assisted stratified sampling.
Then, we provide 4 different stratification optimization
methods with varying trade-offs between optimization overhead and the
aggressiveness of optimization. In the end, we describe our implementation
details in PostgreSQL with an aggregate B-tree implementation,
AB-tree~\cite{10.14778/3538598.3538606}.

\subsection{Two-phase index-assisted approximate query evaluation framework}


\begin{algorithm}[t]
    \small
\caption{\small Two-phase index-assisted approx. query evaluation}
\label{algo:pswr_overview}
    \KwIn{$\mathcal{Q}$: original query;\\
    $\varepsilon, 1 - \delta$: desired half-width of confidence interval and level}
    $n_{0}$: initial sample budget; \\
\KwOut{periodic unbiased estimators $\tilde{A}$ and half-width of confidence intervals $\varepsilon$}
    \Comment{Phase 0} 
    $\mathcal{S}_0 \gets \gamma_{n_0} \sigma_{\mathcal{P}_r} T$ \;
    $\vec{B}, \vec{\sigma}, \vec{h} \gets \texttt{optimize}(\mathcal{S}_0)$ \label{line:optimize}\;

    $\tilde{A}_0, \varepsilon_0 \gets \tilde{A}(\mathcal{S}_0), \varepsilon(\mathcal{S}_0)$\;
    Output $\tilde{A}_0, \varepsilon_0$ \;

    \Comment{Phase 1}
    $S \gets \phi$; $n \gets 0$\;
\While {$\varepsilon_1 \geq \varepsilon$}{
    Compute next sample batch size and sample size allocation  $\Delta n, \vec{n_i}$ based on modified Neyman allocation \label{line:invoke_neyman_allocation}\;
    $n \gets n + \Delta n$ \;
    $S \gets S \cup \bigcup_{1 \leq j \leq k} \Gamma_{n_j}\sigma_{x \in [B_j, B_{j+1})} T$\;
    $\tilde{A}_1, \varepsilon_1 \gets \tilde{A}'(\mathcal{S}), \varepsilon'(\mathcal{S})$\;
    Output 
    $\tilde{A} = \frac{n_0 \tilde{A}_0 + n \tilde{A}_1}{n_0 + n}$ and
        $\varepsilon = \frac{n_0^2 \varepsilon_0 + n^2
        \varepsilon_1}{(n_0 + n)^2}$ \label{line:combined_estimator}\;
}
\end{algorithm}

Algorithm~\ref{algo:pswr_overview} shows our two-phase index-assisted approximate query
evaluation framework, where phase 0 performs uniform sampling and phase 1
performs index-assisted stratified sampling using the stratification derived
from phase 0 samples. A user specifies an ad-hoc query $\mathcal{Q}$ in the
form defined in Equation~\ref{eqn:original_query}, the desired half-width of
confidence interval $\varepsilon$ and confidence level $\delta$.  In addition,
we require the user to specify an initial sample size $n_0$. This is
proportional to the time the user allows the system to spend to derive an
optimal stratification -- in the case where intra-query-range estimator
variance is so small that $n_0$ is sufficient to achieve the desired confidence
bound, the system will skip phase 1 entirely.

When phase 0 cannot achieve the desired confidence bound, we use one of our
stratification optimization methods (to be described in the next subsection),
to produce three vectors (line~\ref{line:optimize}): the stratum boundaries
$\vec{B} = [L, x_1, \ldots x_{k - 1}, U]$ (recall that $[L, U)$ is the query
range in $\mathcal{P}_r$), the estimated per-stratum estimator standard deviation
$\vec{\sigma} = [\sigma_1, \ldots, \sigma_k]$, and the estimated LCA
height/per-sample cost $\vec{h} = [h_1, \ldots, h_k]$. Then we repeatedly
perform index-assisted stratified sampling with the stratification until the
desired confidence bound is achieved.  In each loop, we compute the next sample
batch size and sample size allocation using the modified Neyman allocation
(Algorithm~\ref{algo:modified_neyman_allocation}), which also considers that
each estimator we output in the loop combines both phase 0 and phase 1
estimators (line~\ref{line:combined_estimator}). Note that we make two small
changes. (1) For online aggregation, it is desirable to periodically produce answers
so it is set to the smaller of a pre-configured \texttt{step-size} (or it can be set to
infinity if the user wants the execution to finish as fast as possible) and
the estimated remaining sample size needed to achieve the confidence bound; (2)
For each stratum, in order for the Central Limit Theorem to be valid, we follow
the suggestions in the literature~\cite{haas97:_large} to ensure at least $30$ samples are fetched
from each stratum, if its allocated sample size is smaller than $30$ in modified Neyman allocation.

\begin{algorithm}[t]
    \caption{\small Modified Neyman allocation (line~\ref{line:invoke_neyman_allocation} in
    Algorithm~\ref{algo:pswr_overview})}
    \label{algo:modified_neyman_allocation}
    \KwIn{
    $\vec{\sigma}, \vec{h}$: per-stratum estimator variance and sampling cost\\
    $n_0$: phase 0 sample size \\
    $\varepsilon_0$: phase 0 estimated confidence interval \\
    $\varepsilon, \delta$: desired confidence bound
    }
    \KwOut{
    $n$: next batch sample size \\
    $\vec{n}$: sample size allocation}
    
    $k \gets |\vec{\sigma}|$ \;
    $\sigma^2 \gets \sum_{j=1}^k \sqrt{h_i}\sigma_i * \sum_{j=1}^k \frac{\sigma_i}{\sqrt{h_i}}$ \;
    $t1 = \frac{Z_\delta^2\sigma^2}{2\varepsilon^2} - n_0$ \;
    $t2 = t1^2 + n_0^2(\varepsilon_0^2/\varepsilon^2 - 1)$ \;
    $n_{tot} \gets \min\{t1 + \sqrt{t2}, \texttt{step-size}\}$\;
    \For{$j \gets 1 $ to $k$}{
        $n_j \gets \max\{30, \frac{\sigma_i/\sqrt{h_i}}
        {\sum_{i=1}^k\sigma_i/\sqrt{h_i}} n_{tot}\}$
    }
    Output $n = \sum_{i = 1}^k n_j, n = [n_1, n_2, \ldots, n_k]$ \;
\end{algorithm}


\begin{table}[t]
	\centering
    \vspace{-10pt}
    \caption{\reva{Summary of stratification strategies in this paper}}
	\vspace{-10 pt}
    \footnotesize
\reva{
\begin{tabular}{|l|l|l|l|l|}
\hline
    Strategy & Effectiveness & Extension? (\S~\ref{sec:related})  & Overhead \\ \hline
    \texttt{Greedy}         & Mostly Good & No& Medium \\ \hline
    \texttt{CostOpt}        & Good & Yes & Higher \\ \hline
    \texttt{SizeOpt}        & Varies a lot & Yes & Higher \\\hline
    \texttt{Equal}          & Varies a lot & Yes & Smallest  \\ \hline
\end{tabular}
    }
	\label{tab:summary_strategies}
\end{table}

\vspace{-5pt}
\subsection{Stratification optimization}
\label{sec:stratification}

In this subsection, we present 4 different methods for stratification
optimization with varying tradeoff between optimization overhead and sampling
cost reduction. Note that stratification optimization in \sysname is executed
during query execution and thus is counted towards the entire query execution
latency -- hence, we must strike a balance between bounding optimization
latency and how much we can improve future sampling cost.

\reva{As an overview, we summarize the trade-offs between these 4 different
methods in Table~\ref{tab:summary_strategies}. In general, the first two
methods \texttt{Greedy} and \texttt{CostOpt} have similarly good
effectiveness.
The trade-off is \texttt{CostOpt} has slightly higher overhead but
tends to optimal cost if given sufficient optimizer time while
\texttt{Greedy} may not always produce the best solution. Note that \texttt{Greedy} requires using the
physical sampling index to guide the optimization, and thus is limited to the
single-table aggregation we discuss so far. The techniques in \texttt{CostOpt},
however, are extendible beyond single tables, which we will discuss in
Section~\ref{sec:related}. The other two methods are baseline solutions based
on existing works in the literature, and they generally perform poorly and
cannot achieve optimal cost.}

\subsubsection{Greedy.}
\label{sec:greedy}

\begin{algorithm}[t]
    \small
\caption{\small \texttt{Greedy} Stratification Optimization}
\label{algo:pswr_tree}
    Skip line 2 in Algorithm~\ref{algo:pswr_overview} \;
    Partition $D = \cup \mathcal{D}$ using AB-tree and each stratum $D' \in
    \mathcal{D}$ is the full key range of some subtrees on the left-most or
    right-most paths for the query range $D$\;
    For each $D' \in \mathcal{D}$, $S_{D'} \gets \Gamma_{\Delta n_0}\sigma_{x \in D} T$ \;
    $\mathcal{S}_0 \gets \bigcup_{D' \in \mathcal{D}}S_{D'}$ \;
    $n_0 \gets n_0 - \Delta n_0 \times |\mathcal{D}|$ \;
    Compute current cost $c = c_{opt}$ as in Equation~\ref{eqn:cost_optimization_2} \;
    \While{$n_0 \geq 0$} {
        Find $D' \in \mathcal{D}$ with the largest estimator variance $\sigma^2_{D'}$ \;
        Remove $D'$ from $\mathcal{D}$\;
        $\Delta k \gets $ number of children of $D'$ \;
        Add all subtrees rooted at $D'$'s children as new strata \;
        For each new strata $D'$, $S_{D'} \gets \Gamma_{\Delta n_0}\sigma_{x \in D} T$ \;
        $\mathcal{S}_0 \gets \mathcal{S} \cup \bigcup_{\textrm{new stratum} D'}S_{D'}$ \;
        $n_0 \gets n_0 - \Delta n_0 \times \Delta k$ \;
        Recompute cost $c' = c_{opt}$ as in Equation~\ref{eqn:cost_optimization_2} \;
        \If(\tcp*[h]{\small $\tau$ is stopping threshold}){$(c - c') / c < \tau$}{ 
            \Break;
        }
        $c \gets c'$ \;
    }
    Replace line 4 of Algorithm~\ref{algo:pswr_overview} with modified
    estimators for samples from strata with overlapping (see Section~\ref{sec:greedy})\;
    Output $\vec{B}, \vec{\sigma}, \vec{h}$ for all strata in $\mathcal{D}$ \;
\end{algorithm}

%

\begin{figure}[t]
	\vspace{-10pt}
	\begin{minipage}[b]{0.48\linewidth}
        \centering
        \includegraphics[width=\textwidth]{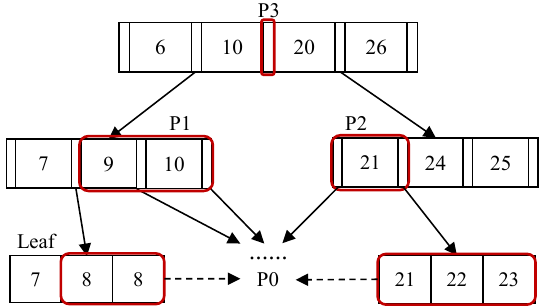}
		\vspace{-16pt}
        \caption{Greedy initial stratification for query range $[8, 24)$}
        \label{fig:structurepar}
    \end{minipage}
    \hfill
	\begin{minipage}[b]{0.48\linewidth}
        \centering
        \includegraphics[width=0.78\textwidth]{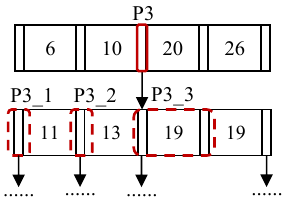}
		\vspace{-8pt}
        \caption{Dividing a stratum with the highest variance}
        \label{fig:despar}
    \end{minipage}
\end{figure}

We present a top-down structure-guided greedy algorithm for stratification
optimization (Algorithm~\ref{algo:pswr_tree}).  We first opportunistically
partition the entire range into a number of strata where each stratum is
covered by some subtrees on the left-most or right-most path found during the
preprocessing for query range $D$ (line 1). For example,
Figure~\ref{fig:structurepar} shows the leftmost and rightmost paths for query
range $[8, 24)$ in a 3-level aggregate B-tree, we can then easily identify the
following strata: (1) all subtrees between the two paths in the LCA
(\texttt{P3}); (2) subtrees in all internal nodes on the left-most or
right-most paths whose subtrees' key ranges are entirely covered by the query
range (\texttt{P1} and \texttt{P2}); (3) leaf entries on the left-most and
right-most paths covered by the query range (\texttt{P0} -- as an optimization,
we exactly aggregate this stratum instead of sampling as it is small). This
step is essentially free because it requires no additional index page access
and can always reduce future sampling cost thanks to
Theorem~\ref{thm:partition_is_better_general}.

To further reduce cost,
we start picking existing stratum with the highest potential cost reduction
and greedily break it down into a number of strata corresponding to the immediate subtrees
rooted at its children. To do so, instead of drawing $n_0$ samples upfront from
the entire query range $D$ in Algorithm~\ref{algo:pswr_overview} line 2, we
draw a smaller sample of size $\Delta n_0$ from each stratum, which we call
the \textbf{per-stratum sample size} and is \texttt{Greedy}'s parameter, and compute
the estimated total cost to reach the desired confidence bound using
Equation~\ref{eqn:cost_optimization_2} under the current stratification. Then,
we greedily pick one stratum $D'$ with the largest estimator variance and break it down
into a number of strata corresponding to subtrees rooted at $D'$'s children (see Figure~\ref{fig:despar} for an example).
However, as an heuristic optimization, if two or more adjacent partitions only
contain the same keys (which are common for low-cardinality dimension columns
in a large table, such as date), we do not further partition them.
We repeat this process until the relative improvement of overall cost is smaller
than a preset \textbf{stopping threshold} $\tau$ or we run out of the initial
sample budget $n_0$.

Another caveat in \texttt{Greedy} is we need to account for samples from overlapping
strata when computing the initial estimator and half-width of confidence
interval (line 19). Conceptually each sample can be tagged with which stratum
it is sampled from. When there are overlapping strata spanning two adjacent levels, the overlapping
must stem from the greedy stratification -- i.e., there's a parent stratum
and a number of child strata which collectively match with the range of the parent
stratum. Suppose there are $\Delta k$ such child strata, then we take a simple arithmetic
mean of all the estimators $\tilde{A}$ derived from the $\Delta k + 1$ strata, which is still unbiased,
and treat it as if it were derived for the parent stratum. The confidence interval $\varepsilon$
also needs to be appropriately scaled by $(\Delta k + 1)^2$. In our actual implementation,
we do not need to physically tag the samples. Instead, we can perform streaming
aggregation with appropriate scaling factors (more details in supplemental materials),
with no extra overhead.

\texttt{Greedy} is a heuristic approach that runs very fast (for only
performing a linear amount of additional computation to the number of strata).
However,
it may run out of budget before it is able to break down a high-variance
stratum deep in the tree, and thus could sometimes result in inferior stratification
than the next method below.

\subsubsection{CostOpt.}
\label{sec:costopt}

\begin{algorithm}[t]
    \small
    \caption{\small \texttt{CostOpt} Stratification Optimization}
    \label{algo:dp}

    Sort $\mathcal{S}_0$ in range key order\;
    $\vec{C}, \vec{m}, \vec{S}, \vec{S2}, \vec{h} \gets$ collect candidate
    boundaries, cumulative moment statistics, and cumulative LCA heights (Def.~\ref{prop:moment_statistics})\;
    $K \gets |key|$\;
    $g_1 \gets [\sigma[L, C_i)\sqrt{h[L, C_i)}$ for $0 \leq i \leq K]$ \;
    $\texttt{min}f \gets r +
    \frac{Z_\delta^2}{\varepsilon^2}g_1^2[K]$ \;
    \For{$k = 2$ to $K$}{
        compute $g_k[K]$ (Equation~\ref{eqn:recursion_g}) \;
        $\texttt{min}f' \gets kr +
        \frac{Z_\delta^2}{\varepsilon^2}g_k^2[K]$ \;
        \If{$\texttt{min}f' \geq \texttt{min}f$}
        {
            $k \gets k - 1$ \; 
            \Break;
        }
        compute $g_k[k..K-1]$ (Equation~\ref{eqn:recursion_g})\;
    }
    $B \gets \argmin_{B} g_k[K]$ \;
\end{algorithm}

Next we present a bottom-up method, \texttt{CostOpt} 
(Algorithm~\ref{algo:dp}). At a glance, we first use the initial sample $\mathcal{S}_0$ of
size $n_0$ to obtain a set of candidate boundary keys -- for now, we assume
these are distinct keys from the sample $\mathcal{S}_0$ and the query range's
lower and upper bounds L and U. We sort and store this set in $\vec{C} = [C_0 =
L, C_1, C_2, \ldots, C_{K} = U]$ of size $K + 1$. Then we minimize $c_{opt}$ in
Equation~\ref{eqn:cost_optimization_2}
using a dynamic programming algorithm below.

Before we go into the details, we note that we must be able to
calculate the estimator variance and average per-sample cost\footnote{A stratum
for arbitrary range could span multiple subtrees at different heights on the
left-most and right-most path during preprocessing in aggregate B-tree. We
adopt another optimization: each sample only needs to start
from a lower subtree if it falls into a subtree is on the left-most or
right-most paths.
Hence, the per-sample cost can be an average instead of just the height of
LCA.} for any subrange with boundary keys from $\vec{C}$, in order for
calculating $c_{opt}$. While these can be re-estimated using $\mathcal{S}_0$ on
demand, this will increase the optimization cost by a factor of $n_0$ -- which
could be too expensive. Instead we maintain a few cumulative statistics in linear time and
space to be able to compute them in $O(1)$ time.
More specifically, let $0 \leq j' < j \leq K$ be two
indices into $\vec{C}$, and we denote the estimator variance for range $[C_{j'},
C_{j})$ (if it is used as a stratum) as $\sigma^2[C_{j'}, C_{j})$ and
per-sample cost as $h[C_{j'}, C_{j})$. We maintain four cumulative vectors
$\vec{m}, \vec{S}, \vec{S2}, \vec{h}$ as follows:

\vspace{-10 pt}
{\allowdisplaybreaks
\begin{align*} 
    m_j &= |\{t \in \mathcal{S}_0 | t.x < C_j\}| \\
    S_j &= n_0\sum_{t \in \mathcal{S}_0 \wedge t.x < C_j} e(t)\mathcal{P}_f(t)\\
    S2_j &= \sum_{t \in \mathcal{S}_0 \wedge t.x < C_j}
    (n_0 e(t)\mathcal{P}_f(t) - S_j/m_j)^2 \\
    h_j &= \sum_{t \in \mathcal{S}_0 \wedge t.x < C_j} \texttt{LCA height of t} \\
\end{align*}
}
\vspace{-24 pt}

\begin{proposition}
    \label{prop:moment_statistics}
    \vspace{-7pt}
    It takes $O(1)$ time to compute any $\sigma^2[C_{j'}, C_j)$ and
    $h[C_{j'}, C_j)$, with $O(n_0)$-time pre-processing and $O(K)$ space\footnote{
This is based on the classic Youngs-Cramer
algorithm~\cite{doi:10.1080/00401706.1971.10488826}, a
numerically stable algorithm in PostgreSQL for computing variance and
standard deviation. We explain this using uniform sampling for
the ease of illustration, but it can be extended to weighted
sampling by substituting $m_i$ with empirical cumulative distribution
functions and appropriately scaling of $S$ and $S2$.}.
\vspace{-7pt}
\end{proposition}
\begin{proof}
    Preprocessing and storage of $\vec{m}, \vec{S}, \vec{S2}_j$ and $\vec{h}$ obviously
    take $O(K)$ space and $O(n_0)$ time after sorting. $\forall 0 \leq j' < j \leq K$, the estimator variance for
    partition $[C_{j'}, C_j)$ can be estimated as follows: 
    \begin{align*}
        \phantom{1} \sigma^2[C_{j'}, C_j) =& (S2_j - S2_{j'-1} -  \\
        \phantom{1} & \frac{m_{j'-1}(m_j -
	    m_{j'-1})(\frac{S_{j'-1}}{m_{j'-1}} - \frac{S_j - S_{j'-1}}{m_j -
	    m_{j'-1})})^2}{m_j}) \\ &\phantom{1} \times \frac{m_j - m_{j' - 1}}{(m_j - m_{j'-1} - 1)
         n_0}
    \end{align*}

    And the average LCA height can be estimated as follows:
    \begin{align*}
        h[C_{j'}, C_j) = \frac{h_j - h_{j' - 1}}{m_j - m_{j'-1}}
    \end{align*}

    Both have a constant complexity.
\end{proof}

Now, if we rewrite Equation~\ref{eqn:cost_optimization_2} as a function
over $\vec{B}$, we have
\vspace{-5 pt}
\begin{align*}
    c_{opt}(\vec{B}) &= c_0(|\vec{B}| - 1)
    + \frac{Z_\delta^2}{\varepsilon^2}(\sum_{j = 1}^k
    \sigma[B_{j-1}, B_j)\sqrt{h[B_{j-1}, B_j)})^2
\end{align*}
\vspace{-5pt}

\begin{figure}[t]
    \centering
    \includegraphics[width=.78\linewidth]{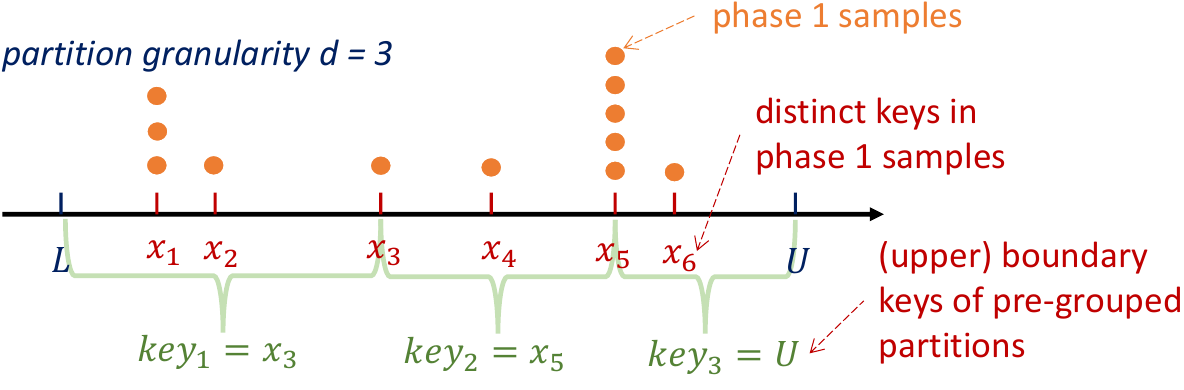}
	\vspace{-14pt}
    \caption{Partition granularity $d = 3$ on a phase 1 sample with
    $6$ distinct values results in $3$ pre-grouped partitions.}
    \vspace{-5pt}
    \label{fig:partition_granularity}
\end{figure}

While the total number of $\vec{B}$ is exponential in $K$, we can derive a
dynamic programming algorithm that works in $O(K^3)$ in the worst-case. As
analyzed previously, $c_{opt}$ is a V-shaped function over the number of
strata, denoted as $k$.  Therefore, we can first optimize the target function
$c_{opt}$ by performing linear search with increasing $k$, and then solve for
minimum of the second term in $c_{opt}$ for each $k$. We can stop as early as
we find the turning point (lines 9 - 11).

To find the minimum of $c_{opt}$ given $k$, it is sufficient to find the minimum of the summation
inside the square -- we denote that as

$$g_k = \min_{\vec{B}: |\vec{B}| = k + 1 \wedge B_0 = L \wedge B_k = U } 
\sum_{1 \leq j \leq k} \sigma[B_{j - 1}, B_j)\sqrt{h[B_{j-1}, B_j)}$$

Slightly overloading the notation, we define $g_k[j]$ as the minimal value of
the summation with $k$ partitions in the range of $[L, B_j)$ (and thus $g_k =
g_k[k]$). Then we can recursively find $g_k[j]$ as follows:
\begin{equation}
    \label{eqn:recursion_g}
    g_k[j] = \min\{g_{k-1}[j'] + \sigma[B_{j'}, B_j)\sqrt{h[B_{j'}, B_j)} | k - 1 \leq j' < j \}
\end{equation}

Now that any $\sigma[B_{j'}, B_j)$ and $h_[B_{j'}, B_j)$ can be computed in $O(1)$
time, each $g_k[j]$ can be computed in $O(K)$ time through dynamic
programming with only $O(K)$ space (since only the latest $g_k$ needs
to be stored at a time if we compute in decreasing order on line 12),
resulting in a total time of $O(K^3)$. 
Finally, we compute the optimal partition boundaries $B$ on line 13
with $O(K)$ time and $O(K^2)$ space, which leverages a
standard technique for tracking optimal decisions in dynamic
programming and we omit for brevity.

Hence, Algorithm~\ref{algo:dp} runs in $O(K^3)$ time and $O(K^2)$ space. Clearly,
this could be very expensive if the initial sample size is large and the range
predicate column $x$ has a high cardinality. On the flip side, the algorithm
can minimize the future cost. To further limit the optimization overhead and
balance the trade-off between cost and benefit, we introduce one additional
hyper-parameter $d$, the \textbf{partition granularity}
(Figure~\ref{fig:partition_granularity}), such that the optimization will run
in $O(\min\{K^3, d^3\})$ time. Specifically, we group the distinct values in
$\mathcal{S}_0$ into up to $d$ groups with the same number of distinct values in
each, and treat them as up to $d$ minimal strata that we do not further divide
during stratification. As Theorem~\ref{thm:partition_is_better_general}
suggests, this results in higher future sampling cost required, but meanwhile,
bounds the optimization time by a constant. On one extreme when the partition
granularity is $d = +\infty$, we place no restriction on how fine we may
partition the query range.  On the other extreme when the partition granularity
is $d = 1$, we effectively disable the stratification optimization. Later, we
will show in our experiment that a moderate value of $d$ of tens to hundreds
typically works the best.

\vspace{-5 pt}
\subsubsection{SizeOpt.}
\label{sec:sizeopt} We also consider a Neyman allocation only
approach to minimize sample size as an alternative. Based on
Theorem~\ref{thm:partition_is_better_general}, we should partition the query
range into the finest strata that we can achieve.  To do so, we simply collect
all distinct keys from the initial sample $\mathcal{S}_0$, sort them, and
create a different strata with every two adjacent keys as lower and upper bounds.
Then we run the standard Neyman allocation (Lemma~\ref{lem:neyman_allocation}).
As Neyman allocation can be implemented using one single pass over sorted data -- we can compute the estimator variances of each strata $\sigma^2_i$ in a streaming fashion over sorted samples as needed -- we
do not have to additionally maintain any statistics or data structure as we do
in \texttt{Greedy} and \texttt{CostOpt}. Thus, \texttt{SizeOpt} can be cheaper
at the cost of sometimes worse future sampling cost, especially when the per-strata sampling
cost has a larger variance.

\subsubsection{Equal.}
\label{sec:equal}
The last method, \texttt{Equal}, is a baseline approach similar
to~\cite{10.1145/3183713.3196905} that does not require any apriori statistics.
Strata are derived in the same way as \texttt{SizeOpt}. Then we
do not perform any computation of estimator variances -- instead, we simply
allocate equal number of samples to all strata. Note that, the main differences
between \texttt{Equal} and the stratified sampling
in~\cite{10.1145/3183713.3196905}  (which we denote as \texttt{ScanEqual}) are:
(1) \texttt{Equal} uses samples to create strata while \texttt{ScanEqual}
scans the entire table for distinct keys to create strata; (2) \texttt{Equal}
uses sampling indexes to perform index-assisted stratified sampling while \texttt{ScanEqual}
performs scans to perform Bernoulli sampling within each stratum.

\subsection{Implementation of \sysname in PostgreSQL}
\label{sec:impl}

We implemented \sysname in PostgreSQL 13.1 with AB-tree (an aggregate B-tree
implementation)~\cite{10.14778/3538598.3538606}, based on an index-assisted
S-AQP PostgreSQL plugin, pgAQP~\cite{10.14778/3611540.3611602}.  The
original pgAQP plugin only supports the approximate query evaluation
with index-assisted uniform sampling of user specified sample sizes.
Specifically, it has two special index-assisted sampling access
methods: \texttt{SWRScan} and \texttt{SWRIndexOnlyScan}, implemented
as the extensible \texttt{CustomScan} operator. They use the
sampling index over a selected range predicate column to provide sampled tuples and sampling probabilities to
the upper physical plan tree nodes. Upon receiving a query with a
\texttt{TABLESAMPLE SWR(k)} clause, its optimizer hooks rewrite the
query plan to (1) inject the index-assisted access path for the
sampled table; (2) modify the approximate aggregation and CLT-based
confidence interval functions as aggregation of estimators.

To implement the two-phase index-assisted approximate query evaluation
framework in \sysname, we introduce \texttt{TABLESAMPLE} \texttt{PSWR} operator
and significantly modified the query rewriting, planning and execution logics
from pgAQP (Figure~\ref{fig:sysarch}(left)). Following shows as an example syntax
for the flight delay query in Figure~\ref{fig:flight_sample}:

{\small
\begin{verbatim}
SELECT APPROX_COUNT(*), APPROX_COUNT_HALF_CI(0.95)
FROM T TABLESAMPLE PSWR(50000, 8000, 0.95)
WHERE date BETWEEN '2001-09-06' AND '2001-09-20'
  AND cancelled = 1;
\end{verbatim}
}

\begin{figure}[t]
    \centering
    \includegraphics[width=.8\linewidth]{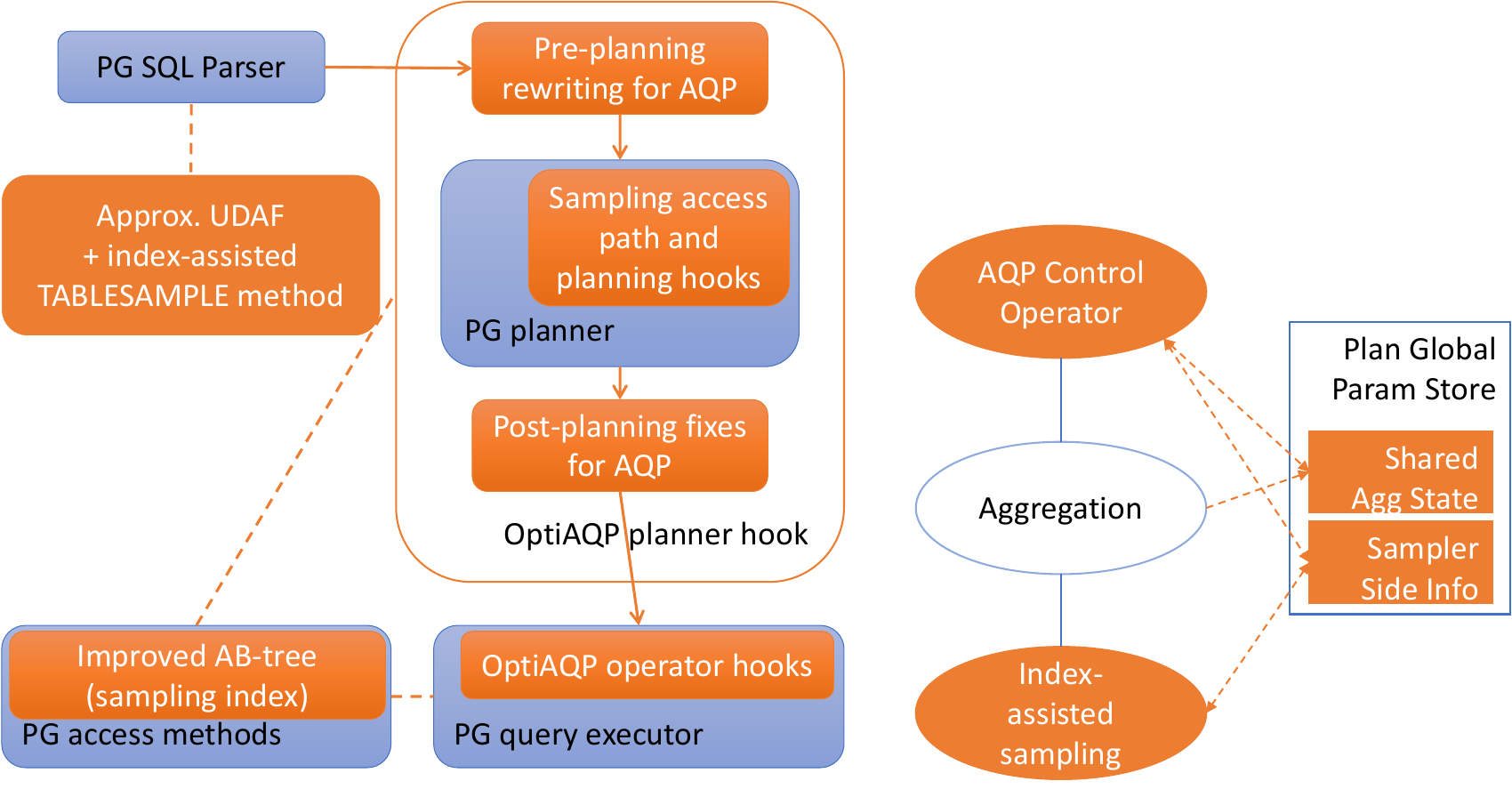}
	\vspace{-13pt}
    \caption{\sysname (left) components; (right) query plan} 
	\vspace{-5pt}
    \label{fig:sysarch}
\end{figure}

Given this query, \sysname will perform the two-phase index-assisted
approximate query evaluation as described in this section, with a target
absolute confidence interval $8000$ at confidence level $0.95$ and an initial
sample size $n_0 = 50000$. The specific stratification method used is
controlled by a session-wide parameter.  Internally, we generate a custom query
plan as shown in Figure~\ref{fig:sysarch}(right). The AQP control operator is a
special operator that serves dual purposes (1) as a standard projection
operator that applies the projection on top of aggregation; (2) to perform the
stratification optimization discussed previously. We further
optimize the implementation to avoid unnecessary recomputation and memory
copies during runtime, by introducing two types of shared states in the plan
global parameter store: (1) sampler side information object, which is used by
AQP control operator to provide new partitioning and sample size information
one partition at a time to the index-assisted sampling operator, as well as the
index-assisted sampling operator to pass sampling probability and filter or
index rejection results back to the aggregation functions; (2) shared
aggregation state object for each unique aggregation expression\footnote{An
approximate aggregation function and the corresponding CI function share the
same aggregation state.} which serves as the sample estimator store for phase
0\footnote{To avoid recomputing the aggregation twice for phase 0, 
we defer the actual aggregation in phase 0 after aggregation
operator returns to the AQP control operator, when we compute the cumulative
moment statistics.}, and the standard streaming aggregation state for phase 1.
In addition, we also modified the original AB-tree with the various
optimization previously discussed.  Our current
implementation only supports single-threaded execution.  Nevertheless, we note
that both phases are embarrassingly parallel, and thus we consider
parallel execution support as a future work.

\cut{

This section details our algorithms of implementation based upon the above motivations. Our implementation includes three parts, including how to allocate samples to each partitions to minimize the cost, how to partition the dataset optimally and how to progressively determine $n_0$.

\subsection{Given partition, allocate samples}

The range predicate $\mathcal{P}_r$ over $t$ can be represented in the form of $t \in [L,U)$. Assuming $\mathcal{P}_r$ has been partitioned to a given set of $k$ disjoint ranges $T_1 = [l_0,l_1), T_2 = [l_1,l_2),..., T_k = [l_{k-1},l_k) $, the question is how to allocate samples to each partitions to minimize the cost.

For each range, the query can be modified as:

\begin{itemize}
    \item[] \texttt{SELECT AGG(e) as $S_{i}^{'}$}
    \item[] \texttt{FROM (SELECT * FROM T WHERE $t_i \in [l_{i},l_{i+1})$) $T_i$}
    \item[] \texttt{WHERE $\mathcal{P}_f$}
\end{itemize}

where $T_i= \Theta_{t_1\in [l_i,l_{i+1})}T=\{r_1, r_2,..., r_{Ni}\}$ and $r_j$ for $l\le j\le N_i$ is the record in T. And then we draw $n_i^{'}$ uniform and independent samples with replacement $\{s_1, s_2,..., s_{n_i^{'}}\}$ in each range in the initial sampling phase, and $S_i'$ is the sum of the value of tuples passed $\mathcal{P}_f$, which is also the target answer. And let $S^{'}=\sum_{i=1}^{k} S_i^{'}$

Set $I(t_{ij})=1$ if tuple $t_{ij}$ satisfies $\mathcal{P}_r$ and $\mathcal{P}_f$, and $I(t_{ij})=0$ otherwise.
$\tilde{A_{n_i}^{'}}\triangleq \frac{\sum_{t_{ij} \in [1,n_i^{'}]}I(t_{ij}) e(t_{ij})/P(t_{ij})}{n_i^{'}}$, which is the estimated answer for each range. 

$E[\tilde{S_i}] = \frac{1}{n_i^{'}}\sum_{t_{ij}}E[I(t_{ij}) e(t_{ij})/P(t_{ij})] = \frac{1}{n_i^{'}}*n_i^{'}*(\sum_{t_{ij} \in \bar{R} \wedge I(t_{ij})=0}0 + \sum_{t_{ij} \in \bar{R} \wedge I(t_{ij})=1} \frac{e(t_{ij})}{p(t_{ij})}p(t_{ij})) = \sum_{t_{ij} \in \bar{R} \wedge I(t_{ij})=1} e(t_{ij}) = A $.

$\sigma_i^2 = Var[\tilde{S_i}]=\frac{1}{n_i^{'}}Var[I(t_{ij}) e(t_{ij})/P(t_{ij})]$


$\widehat{Var}[I(t_{ij}) e(t_{ij})/P(t_{ij})] =  \frac{\sum_{t_{ij}}(I(t_{ij})*e(t_{ij})/P(t_{ij}))^2 - n_i^{'} \tilde{S_i}^2} {n_i^{'} - 1} $

According to central limit theorems (CLT's) for independent and identically distributed (i.i.d.) random variables, $\frac{\sqrt{n_i}(\tilde{N_i'} -\mu_i)}{\sigma_i}\Rightarrow N(0,1)$ as $n_i \to \infty$, so we can assume $\tilde{N_i} \sim N(\mu_i, \frac{\sigma_i^{2}}{n_i})$ as $n_i$ is large enough.

$\tilde{N'}=\sum_{i=1}^{k}\tilde{N_i'}$, we can also assume  $\tilde{N'}\sim N(\sum_{i=1}^{k}\mu_i, \sum_{i=1}^{k} \frac{\sigma_i^{2}}{n_i})$, because $N_i'$s are all independent.

And its confidence interval can be formulated as follows:
$N'\in [\tilde{N'}-\varepsilon,\tilde{N'}+\varepsilon]$ with probability p.

$E[\tilde{N'}]=N'=\sum_{k}^{i=1}N_i'$

$\tilde{N'}-N' \sim N(0, \sum_{i=1}^{k} \frac{\sigma_i^{2}}{n_i})$, let $\sigma^2 = \sum_{i=1}^{k} \frac{\sigma_i^{2}}{n_i}$, so $\frac{\tilde{N'}-N'}{\sigma} \sim N(0, 1)$

So we have $P(|\frac{\tilde{N'}-N'}{\sigma}|\le \frac{\varepsilon}{\sigma})\approx 2\Phi (\frac{\varepsilon}{\sigma})-1=p$, $\Phi (\frac{\varepsilon}{\sigma})=\frac{p+1}{2}$ for large samples.

for $\varepsilon>0$, where $\Phi$ is the cumulative distribution function of an N(0,1) random variable. $Z_p$ is the $(p+1)/2$ quantile of this distribution function, so $\Phi(Z_p) = (p+1)/2$.

So $Z_p = \frac{\varepsilon}{\sigma}, \varepsilon = Z_p\sigma = Z_{p}\sqrt{\sum_{i=1}^{k}\frac{\sigma_{i}^{2}}{n_{i}}}$

Minimize $C=\sum_{i=1}^{k}n_{i}$, subject to $Z_{p}\sqrt{\sum_{i=1}^{k}\frac{\sigma_{i}^{2}}{n_{i}}}=\varepsilon$. According to Lagrange multiplier, $n_i$ can be estimated as follows:
\begin{equation}
n_i=\frac{Z_{p}^{2}\sigma_{i}\sum_{i=1}^{k}\sigma_{i}}{\varepsilon^{2}}
\end{equation}

so the total cost of the query is $C \ge \frac{Z_{p}^{2}
(\sum_{i=1}^{k}\sigma_{i})^{2}}{\varepsilon^{2}}$, where $\sigma_{i} = \sqrt{\frac{\sum_{t_{ij}}(I(t_{ij})*e(t_{ij})/P(t_{ij}))^2 - n_i \tilde{S_i}^2} {n_i^{'} - 1}}$

\subsection{How to get partitions}

See Algorithm~\ref{algo:dp}. Clustering can be used to divide M nodes into K clusters through calculating the distance between pairs of clusters. However, the cost cannot be sure to be minimal. As a result, we utilize dynamic programming to find optimal ways to partition the nodes into k clusters to minimize the cost.

In dynamic programming, we need to find how to partition the dataset (find k and where to split) to minimize the cost:

Let $N_{i}^{'}\triangleq\rho_{i}N_{i}$, 

then $C \ge \frac{Z_{p}^{2} (\sum_{i=1}^{k}\sqrt{\frac{\sum_{t_{ij}}(I(t_{ij})*e(t_{ij})/P(t_{ij}))^2 - n_i^{'} \tilde{S_i}^2} {n_i^{'} - 1}} )^{2}}{\varepsilon^{2}}$, 

so the target of the optimization:

$f = \sum_{i=1}^{k}\sqrt{\frac{\sum_{t_{ij}}(I(t_{ij})*e(t_{ij})/P(t_{ij}))^2 - n_i^{'} \tilde{S_i}^2} {n_i^{'} - 1}}$.

\begin{algorithm}[!ht]
\caption{Dynamic Programming}\label{algo:dp}
\begin{algorithmic}[1]
\REQUIRE M points of $(key, )$, the number of clusters K
\ENSURE K clusters

\STATE $t_{ki} \gets +\infty$ for all $k \in [1, K]$ and $i \in [1, M]$ 

\FOR{$i = 1\;to\;M$} 
     \STATE $f_{1i}= \sqrt{\frac{1}{n_i^{'} - 1} * \sum_{t_{ij}}((I(t_{ij})*e(t_{ij})/P(t_{ij}))^2 - n_i^{'} \tilde{S_i}^2)}$
\ENDFOR
\FOR{$k = 2\;to\;K$} 
    \FOR{$i = k\;to\;M$} 
         \STATE $f_{ki} \gets$ $\min \{
         k-1\leq j\leq i - 1 \;|\;
         f_{(k-1)j}+\sqrt{\frac{1}{n_j^{'} - 1} *(\sum_{t_{ij}}(I(t_{ij})*e(t_{ij})/P(t_{ij}))^2 - n_j^{'} *(\tilde{S_{j+1}}^2-\tilde{S_{j}}^2))}\}$
    \ENDFOR
\ENDFOR
\end{algorithmic}
\end{algorithm}

\subsection{Progressively get n0}

}

\cut{ 
\section{CUT BELOW}

===============================================

\begin{theorem}
    \label{def:lyapunov_clt} (Lyapunov CLT~\cite[Thm.~27.3]{alma991026726359704801})

    Suppose $\{X_1, \ldots, X_n, \ldots\}$ is a sequence of
    independent random variables, each with finite expected value
    $\mu_i$ and variance $\sigma_i^2$. Let $s_n^2 = \sum_{i=1}^n
    \sigma_i^2$. If for some $\delta > 0$, $\lim_{n \rightarrow
    \infty} \frac{\sum_{i=1}^n E[|X_i - \mu_i|^{2 + \delta}]}{s_n^{2 + \delta}} = 0$,
    then $\frac{\sum_{i=1}^{n}{X_i - \mu_i}}{s_n}$ converges to
    a standard normal distribution.
\end{theorem}

For a dataset A with $M$ distinct keys and $N$ tuples, we sort all the keys in increasing order, numbered as from 1 to $M$. We consider the following query formulation in this research:
\begin{itemize}
    \item[] \texttt{SELECT AGG(attribute collection)}
    \item[] \texttt{FROM A}
    \item[] \texttt{WHERE p1 and p2}
\end{itemize}

where AGG is one of the standard aggregation functions, such as SUM, COUNT, AVG, and p1 and p2 are two predicates, where p1 is a range predicate over some column $t_1$, in the form of $t_1 \in [L, U)$, and p2 can be an arbitrary predicate. For the rest of the formulation, we first consider COUNT queries and will discuss SUM and AVG queries later.

We aim to partition the range $t_{1}\in [L,U)$ of table A into a set of $k$ disjoint ranges $A_1 = [l_0,l_1), A_2 = [l_1,l_2),..., A_k = [l_{k-1},l_k) $ to minimize the cost. 

The total cost can be computed as the following:

For each range, the query can be modified as:

\begin{itemize}
    \item[] \texttt{SELECT COUNT(*) as $N_{i}^{'}$}
    \item[] \texttt{FROM (SELECT * FROM A WHERE $t_i \in [l_{i},l_{i+1})$) $A_i$}
    \item[] \texttt{WHERE p2}
\end{itemize}

where $A_i= \Theta_{t_1\in [l_i,l_{i+1})}A=\{r_1, r_2,..., r_{Ni}\}$ and $r_j$ for $l\le j\le N_i$ is the record in A. And then we draw $n_i$ uniform and independent samples with replacement $\{s_1, s_2,..., s_{n_i}\}$ in each range, and $N_i$ is the number of tuples with the $i^{\texttt{th}}$ key, and $N_i'$ is the number of tuples among them that passes the second predicate $p2$, which is also the target answer. And let $N=\sum_{i=1}^{k} N_i$ and $N^{'}=\sum_{i=1}^{k} N_i^{'}$.

For any $j \in [1, n_i]$, and any $l \in [1, N_i]$, $Pr(s_j=r_l) = \frac{1}{N_i}$

Set $v(i)=1$ if tuple $m_i$ satisfies p2 and $v(i)=0$ otherwise.
$\tilde{N_i}\triangleq \frac{\sum_{j\in [1,n_i]}V(s_j)}{n_i}N_i $, which is the estimated answer for each range. 

$E[\tilde{N_i}] = \frac{1}{n_i}\sum_{j\in[1,n_i]}Pr(p_2(s_j))\cdot N_i=\frac{1}{n_i}\cdot n_i\cdot \frac{N_i'}{N_i}\cdot N_i = N_i'$, since $Pr(p_2(s_j))=\frac{N_i'}{Ni}$.

$\sigma_i^2 = Var[\tilde{N_i}]=\frac{N_i'}{Ni}(N_i-N_i')^2+\frac{N_i-N_i'}{Ni}(0-N_i')^2=N_i'(N_i-N_i')$

According to central limit theorems (CLT's) for independent and identically distributed (i.i.d.) random variables, $\frac{\sqrt{n_i}(\tilde{N_i'} -\mu_i)}{\sigma_i}\Rightarrow N(0,1)$ as $n_i \to \infty$, so we can assume $\tilde{N_i} \sim N(\mu_i, \frac{\sigma_i^{2}}{n_i})$ as $n_i$ is large enough.

$\tilde{N'}=\sum_{i=1}^{k}\tilde{N_i'}$, we can also assume  $\tilde{N'}\sim N(\sum_{i=1}^{k}\mu_i, \sum_{i=1}^{k} \frac{\sigma_i^{2}}{n_i})$, because $N_i'$s are all independent.

And its confidence interval can be formulated as follows:
$N'\in [\tilde{N'}-\varepsilon,\tilde{N'}+\varepsilon]$ with probability p.

$E[\tilde{N'}]=N'=\sum_{k}^{i=1}N_i'$

$\tilde{N'}-N' \sim N(0, \sum_{i=1}^{k} \frac{\sigma_i^{2}}{n_i})$, let $\sigma^2 = \sum_{i=1}^{k} \frac{\sigma_i^{2}}{n_i}$, so $\frac{\tilde{N'}-N'}{\sigma} \sim N(0, 1)$

So we have $P(|\frac{\tilde{N'}-N'}{\sigma}|\le \frac{\varepsilon}{\sigma})\approx 2\Phi (\frac{\varepsilon}{\sigma})-1=p$, $\Phi (\frac{\varepsilon}{\sigma})=\frac{p+1}{2}$ for large samples.

for $\varepsilon>0$, where $\Phi$ is the cumulative distribution function of an N(0,1) random variable. $Z_p$ is the $(p+1)/2$ quantile of this distribution function, so $\Phi(Z_p) = (p+1)/2$.

So $Z_p = \frac{\varepsilon}{\sigma}, \varepsilon = Z_p\sigma = Z_{p}\sqrt{\sum_{i=1}^{k}\frac{\sigma_{i}^{2}}{n_{i}}}$

Minimize $C=\sum_{i=1}^{k}n_{i}$, subject to $Z_{p}\sqrt{\sum_{i=1}^{k}\frac{\sigma_{i}^{2}}{n_{i}}}=\varepsilon$

$n_i=\frac{Z_{p}^{2}\sigma_{i}\sum_{i=1}^{k}\sigma_{i}}{\varepsilon^{2}}$

so the total cost of the query is $C \ge \frac{Z_{p}^{2}
(\sum_{i=1}^{k}\sigma_{i})^{2}}{\varepsilon^{2}} =\frac{Z_{p}^{2} (\sum_{i=1}^{k}\sqrt{(N_{i}-N_{i}^{'})N_{i}^{'}} )^{2}}{\varepsilon^{2}}$

\textcolor{red}{Prob Revision}

In order to do the partition, we have at least two phases to get the tuples. the first phase is to collect the statistic and the second phase is to get the rest tuples from different partitions.

n: total sampling size.

$n_1$: Sampling size used in Phase 1.

$n_2$: Sampling size used in Phase 2, which equals $n-n_1$.

$t_{kj}^{(k)}$: The tuple got in Phase k in partition range j.

In Phase 1, there is $k_1$ range, where $k_1 = 1$. $n_{11} = n_1$, $E[e_1] = Ans$.

In Phase 2, there are $k_2$ range. $n_2 = n - n_1$.

In different partition, the sampling size is $n_{21}, n_{22}, n_{23}, ..., n_{2k_2}$.

$e_{2k_2} = \frac{\sum_{j=1}^{n_{2k_2}}\frac{1}{p_{k_2}^{(2)}}}{n_{2k_2}}$.

$e_2 = \sum_{k=1}^{k_2}e_{2k_2}$, $E[e_2] = Ans$.

So, $E[\frac{\alpha_1e_1+\alpha_2e_2}{\alpha_1+\alpha_2}] = Ans$.

Assuming $\alpha_1 = \frac{n_1}{n}, \alpha_2 = \frac{n_2}{n}$, and it is uniform sampling.

$E[\frac{\alpha_1e_1+\alpha_2e_2}{\alpha_1+\alpha_2}] = \frac{n_1e_1+n_2e_2}{n} = \frac{n_1\sum_{k=1}^{k_1}\sum_{j=1}^{n_{1k}}(\frac{1}{P(t_{kj}^{(1)})}/n_{1k})+n_2\sum_{k=2}^{k_2}\sum_{j=1}^{n_{2k}}(\frac{1}{P(t_{kj}^{(2)})}/n_{2k})}{n}=\frac{\sum_{p=1}^{2}\sum_{k=1}^{k_p}\sum_{j=1}^{n_{pk}}\frac{1}{n_{pk}P_k^{(p)}/n_p}}{n}$.

So the prob needs to be revised as $prob * \frac{n_pk}{n_p}$.

} 


\section{Experiments} \label{sec:exp}

In this section, we empirically evaluate \sysname. 
All code and data generator/links to datasets are available in our open-source
repository. The main objectives of our experiments include:

\vspace{-1pt}
\begin{pkl}
    \item Can two-phase index-assisted approximate query evaluation achieve a
        desired confidence bound faster than existing non-index-assisted S-AQP
        approaches?
    \item What are the overhead/benefit trade-offs among the four different
        stratification optimization methods?
    \item How sensitive are our approaches to the algorithms' parameters?
\end{pkl}

\subsection{Experiment setup}

\noindent\textbf{Methods under comparison:}

\begin{pkl} 
    \item \textbf{Greedy}: \sysname with \texttt{Greedy}
        (Section~\ref{sec:greedy}). Default per-stratum sample size
        $\Delta n_0 = 600$ and stopping threshold $\tau = 0.004$.


    \item \textbf{CostOpt}: \sysname with \texttt{CostOpt}
        (Section~\ref{sec:costopt}).  By default, we set partition granularity
        $d = 100$.
    \item \textbf{SizeOpt}: \sysname with \texttt{SizeOpt} (Section~\ref{sec:sizeopt}).

    \item \textbf{Equal}: \sysname with \texttt{Equal} (Section~\ref{sec:equal}).


   \item \textbf{Uniform}: Baseline index-assisted uniform sampling in
       pgAQP~\cite{10.14778/3611540.3611602}.  As the original pgAQP only takes
       sample size as parameter, we modified its implementation to take a
       requested confidence bound $(\varepsilon, \delta)$ and an initial sample
       size $n_0$ instead. It does not perform any optimization after the
       initial sampling; rather, only runs the phase 1 loop in
       Algorithm~\ref{algo:pswr_overview} until the confidence bound is
       achieved.
    \item \textbf{ScanEqual}:
        Baseline scan-based stratified sampling in
        VerdictDB~\cite{10.1145/3183713.3196905} on PostgreSQL.
        Since we test the online ad-hoc query settings, we include the time for
        refreshing the sample sets before each query\footnote{We thank
        the authors of~\cite{10.1145/3183713.3196905} for generating the
        query templates for stratified sampling with variational subsampling
        for computing CI for TPC-H queries, as it was disabled in its open-source
        repository. We manually created the rewriting for the rest of the
        queries.}. As \texttt{ScanEqual} only takes a sampling rate as
        parameter, we manually tune it for the requested $\varepsilon$
        with high probability.
    \item \textbf{Exact}: Baseline exact query in PostgreSQL
        with heap scan.
\end{pkl}

In all experiments, we configure all AQP \cut{systems} methods to derive
confidence intervals at $1 - \delta = 0.95$ confidence level. We set
preprocessing factor $c_0 = 100$ and initial sample sizes (or maximum initial
sample size for \texttt{Greedy}) $n_0 = min(200 * NDV, 100000)$, where $NDV$ is
the total number of distinct values of in query range -- it can easily be
estimated using existing DBMS statistics in practice.

\Paragraph{Datasets and queries.} \revc{Dataset statistics are in
Table~\ref{tab:summary_table}.}

\begin{enumerate} \item \textbf{flight}: the US airline on-time
            performance dataset~\cite{DVN/HG7NV7_2008}, enlarged to 10x. 
            We use the query in the motivating
            example in Section~\ref{sec:intro}.

        \item \textbf{intel}: the Intel Lab sensor data~\cite{intel},
            enlarged to 1000x.
            We query the total number of
            readings between '2004-02-28' and '2004-04-05', filtered
            by temperature greater than 27 Celsius degree.

        \item \textbf{census}: the 1994 US Census income
            dataset~\cite{census_income_20}, enlarged to 10,000x. 
            We query the number of surveyees who
            worked between 1 (inclusive) and 100 (exclusive) hours per
            week, filtered by income greater than $\$50K$. 

        \item \textbf{lineitem}: We generated the lineitem table from
            the skewed TPC-H benchmark~\cite{skewed_tpch}, with a
            modification to create special date ranges with
            higher delivery delays during the most common ship dates,
            to simulate holiday seasons.
            We generate several different datasets
            with varying scaling factors and number of special ranges.
            We query the total revenue of sales
            ($\texttt{SUM(l\_extendedprice * (1 - l\_discount))}$) in
            the ship date between '1992-01-01' and '1998-12-31',
            filtered by higher delivery delay:\\
            $\texttt{l\_receiptdate - l\_shipdate} > 49$.
\end{enumerate}

\begin{table}[t]
	\centering
	\caption{\revc{Summary of datasets in this paper}}
        \vspace{-10pt}
    \label{tab:summary_table}
    \revc{
	{
		\footnotesize
\begin{tabular}{|l|l|l|l|l|l|}

\hline
        & Columns & Rows & Table in PG & Sampling Index  \\ \hline
Flight   & 30      & 1.19 B & 211.9 GB      & 13.0 GB                \\ \hline
Intel    & 8       & 2.31 B & 147.1 GB      & 25.3 GB                \\ \hline
Census   & 15      & 0.33 B  & 92.0 GB       & 3.6 GB              \\ \hline
Lineitem & 17      & 0.13 B  & 38.8 GB       & 1.5 GB               \\ \hline
\end{tabular}
	}
}
    \vspace{-5pt}
\end{table}

\Paragraph{System configuration.} We perform experiments on a
dual-socket server with Intel Xeon Gold 6330 CPUs clocked at 2.6
GHz (max frequency at 3.1 GHz in turbo mode) and 512 GB DDR4 RAM. The system is
installed with Ubuntu Linux 22.04 LTS and configured with performance scaling policy.
The database files and logs are placed on a Samsung 990 PRO 4TB NVMe SSD.
All experiments are done in single-thread execution in the PostgreSQL
13.1 with AB-tree sampling index
support~\cite{10.14778/3538598.3538606} under snapshot isolation, built using GCC 11.3 with
\texttt{-O3} flag. 
We configure all the databases with 32 GB shared buffer pools, with all
default optimizer flags except for those of parallelism, which are disabled
to ensure single-threaded execution. 
For all databases, we create an AB-tree index
over the range predicate column and sufficiently warm up the internal and
system I/O buffer cache before experiments. We repeat each test for 10 times
to account for randomness, except tests on VerdictDB which are only run for
3 times due to its long execution time. When we plot data points, we also
plot the spread across runs using the standard whisker plots (representing
1.5 times of the difference between 25\% and 75\% quantiles).

\begin{figure*}[t] 
    \centering
    \begin{minipage}{.7\linewidth}
        \includegraphics[width=\linewidth,clip=true,viewport=0 415
        1200 476]{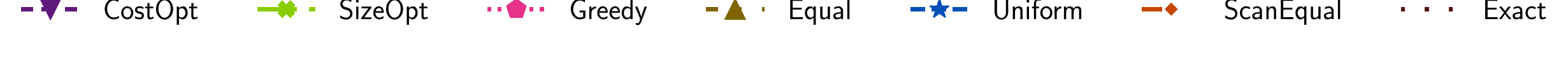}
        \vspace{-20pt}
    \end{minipage} \\
    \subfigure[flight]{
        \includegraphics[width=.27\linewidth]{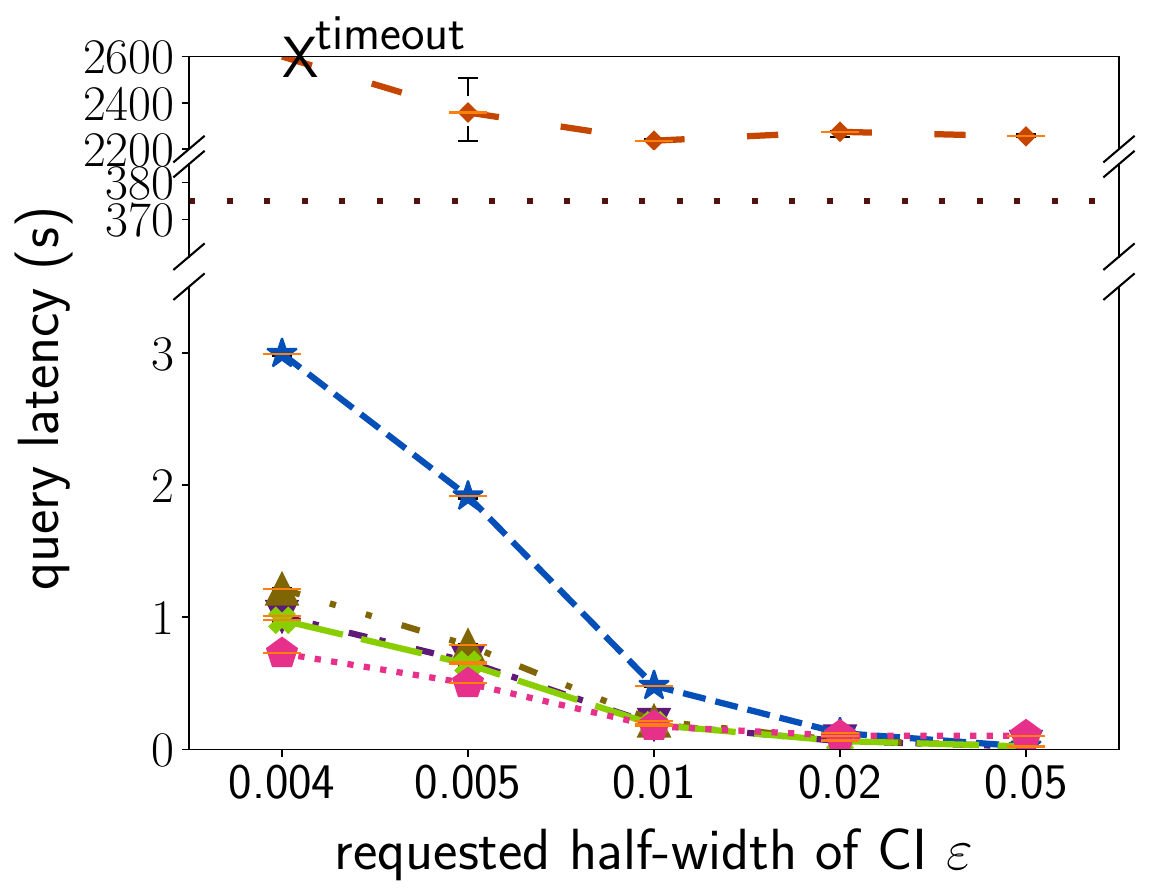}
        \label{fig:flight_latency} 
    }\hfill
    \subfigure[intel]{
        \centering
        \includegraphics[width=.27\linewidth]{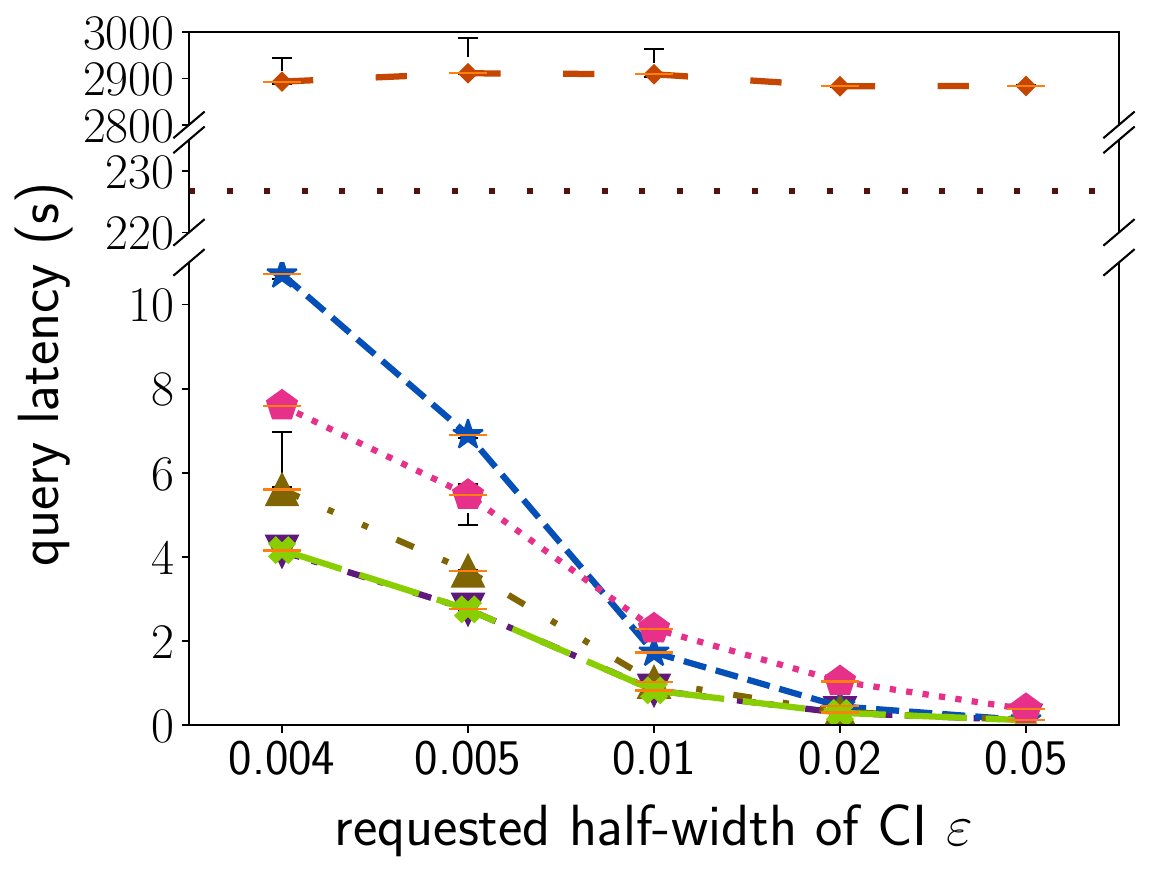}
        \label{fig:intel_latency} 
    }\hfill
	\subfigure[\revc{census}]{
        \centering
        \includegraphics[width=.27\linewidth]{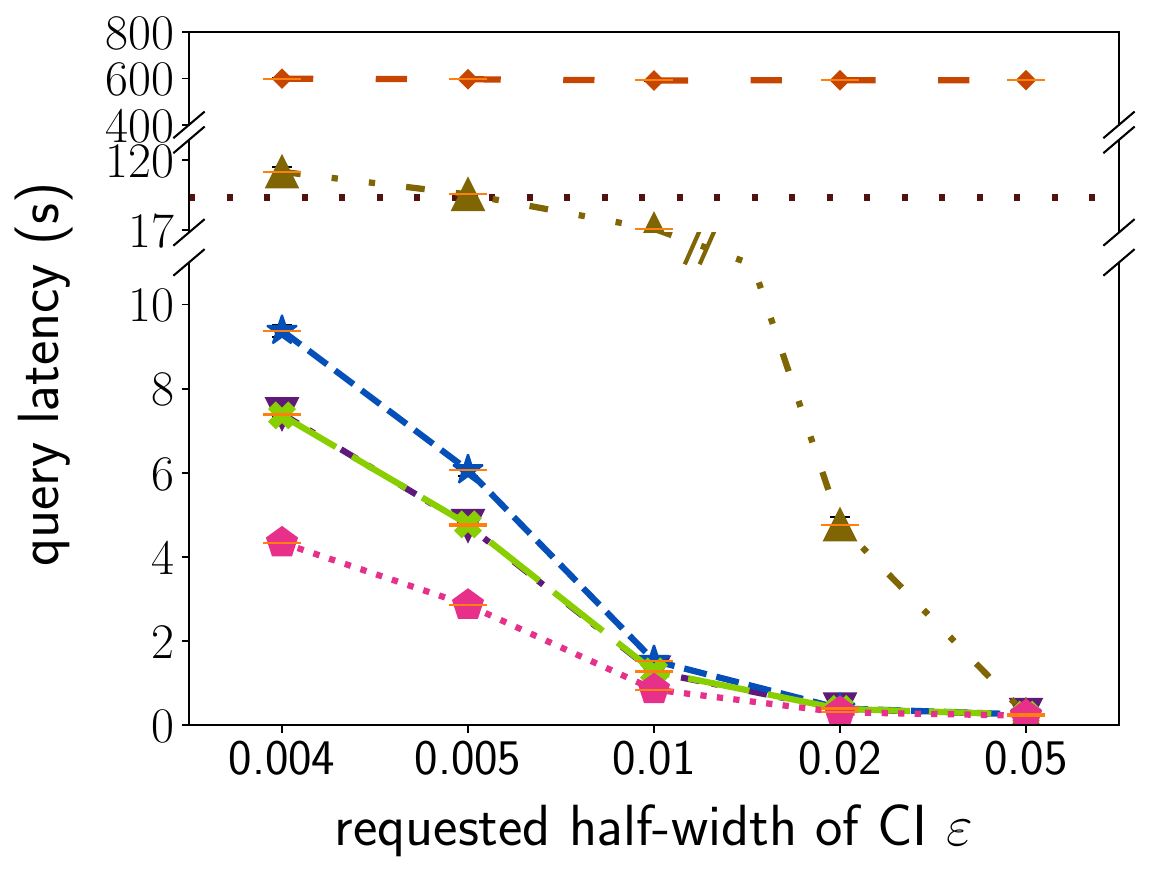}
        \label{fig:census_latency}
    }
    \vspace{-17 pt}
	\caption{\revc{Query latency with varying requested confidence interval}}
    \vspace{-15pt}
    \label{fig:query_latency}
\end{figure*}

\vspace{-8 pt}
\subsection{Latency improvement on real-world datasets}

We first evaluate \sysname's latency improvement over queries on real-world
datasets, which have high value/selectivity skewness.
For each experiment, we first perform exact query in
PostgreSQL, and obtain the exact answers. Then we request a half-width of
confidence interval $\varepsilon$ that corresponds to $0.4\%,0.5\%, 1\%, 2\%,
5\%$ of the exact answers.
Figure~\ref{fig:query_latency}
show the query latencies. Overall, all variants of \sysname outperform
 \texttt{Exact} by 1 - 4 orders of magnitude.
\revc{
Thanks to index-assisted sampling, our methods also significantly outperform
scan-based sampling method \texttt{ScanEqual} by even a larger margin.
For example, the speedup of \texttt{CostOpt}
compared to \texttt{ScanEqual} to achieve the same confidence interval is up to 
98708x on flight dataset, while the speedup of \texttt{CostOpt} over \texttt{Exact}
is up to 16421x on the same dataset.}

In general, \texttt{CostOpt}, \texttt{SizeOpt} and \texttt{Greedy} can
outperform baseline index-assisted sampling \texttt{Uniform} for lower
requested $\varepsilon$ on these datasets and queries. \revc{Specifically, \texttt{CostOpt} 
consistently speeds up queries
when compared to the baseline \texttt{Uniform} by up to 3x in our tested queries.
Compared to the baseline stratification strategy \texttt{Equal}, \texttt{CostOpt}
can achieve up to 14.6x speedup.
\texttt{SizeOpt} performs similarly to \texttt{CostOpt }on
these datasets because the resulting strata's LCAs are roughly at the same
height and thus have similar per-sample cost.} \texttt{Greedy}
performs better than \texttt{CostOpt} and \texttt{SizeOpt} on flight and census
dataset, but worse on intel dataset. The reason is \texttt{Greedy} can find a
good stratification with lower optimization overhead for flight and census, while
it fails to find a good stratification for intel based on the same stopping criteria.
Notably, \texttt{Greedy} even performs worse than the na\"{i}ve \texttt{Equal}
with equal sample size per stratum on intel.
This shows the better robustness of \texttt{CostOpt}.
In addition, we also verified that all methods except \texttt{ScanEqual} are returning
confidence intervals below the requested confidence intervals (figures are omitted
due to space constraint, which are available in the supplemental materials). Note that the error bars for all methods in \sysname
are quite small, indicating that for these high-variance queries, \sysname can
produce quite stable estimations.

\vspace{-10 pt}
\subsection{Scalability}

We evaluate the scalability with increasing data
size.  Specifically, we generate the lineitem table using our modified TPC-H
benchmark, with 3 high shipping delay date ranges and varying scale factors. We
report the query latencies across 10 runs of \texttt{CostOpt},
\texttt{SizeOpt}, \texttt{Greedy}, \texttt{Equal}, \texttt{Uniform}, 
and exact
query in Figure~\ref{fig:scalability}. We observe that \texttt{CostOpt} and
\texttt{SizeOpt} consistently run faster than \texttt{Uniform} under all scale
factors. Notably, at scale factor $60$, \texttt{CostOpt} starts to outperform
\texttt{SizeOpt}, because the LCA heights of strata start to vary on larger
datasets.  Except for scale factor 10, \texttt{Greedy} outperforms
\texttt{Uniform}, and the trend further increases when the scale factor
increases.  The reason is that there are more records associated with the high
shipping delay date ranges when the scale factor increases, leading to even
higher estimator variance without partitioning.  \texttt{Greedy} is able to
optimize the estimator variance through an optimized partitioning to 
speedup of about 3.4x in query latency at scale factor 60. In contrast,
\texttt{Equal} shows a rapid increase in latency at scale factor 60, indicating
that equal allocation without considering the data distribution of the large
dataset is inferior to optimized stratification. Exact query latency, as
expected, increases super-linearly as scale factor increases, which is almost
an order of magnitude higher than \texttt{CostOpt} and \texttt{Greedy} at scale
factor 60. This experiment demonstrates good scalability of \texttt{CostOpt}
and \texttt{Greedy}. 

\vspace{-8 pt}
\subsection{Impact of varying estimator variances}


\begin{figure}[t]
    \centering
    \begin{minipage}{.7\linewidth}
        \includegraphics[width=1.5\linewidth,clip=true,viewport=0 400 900 455]{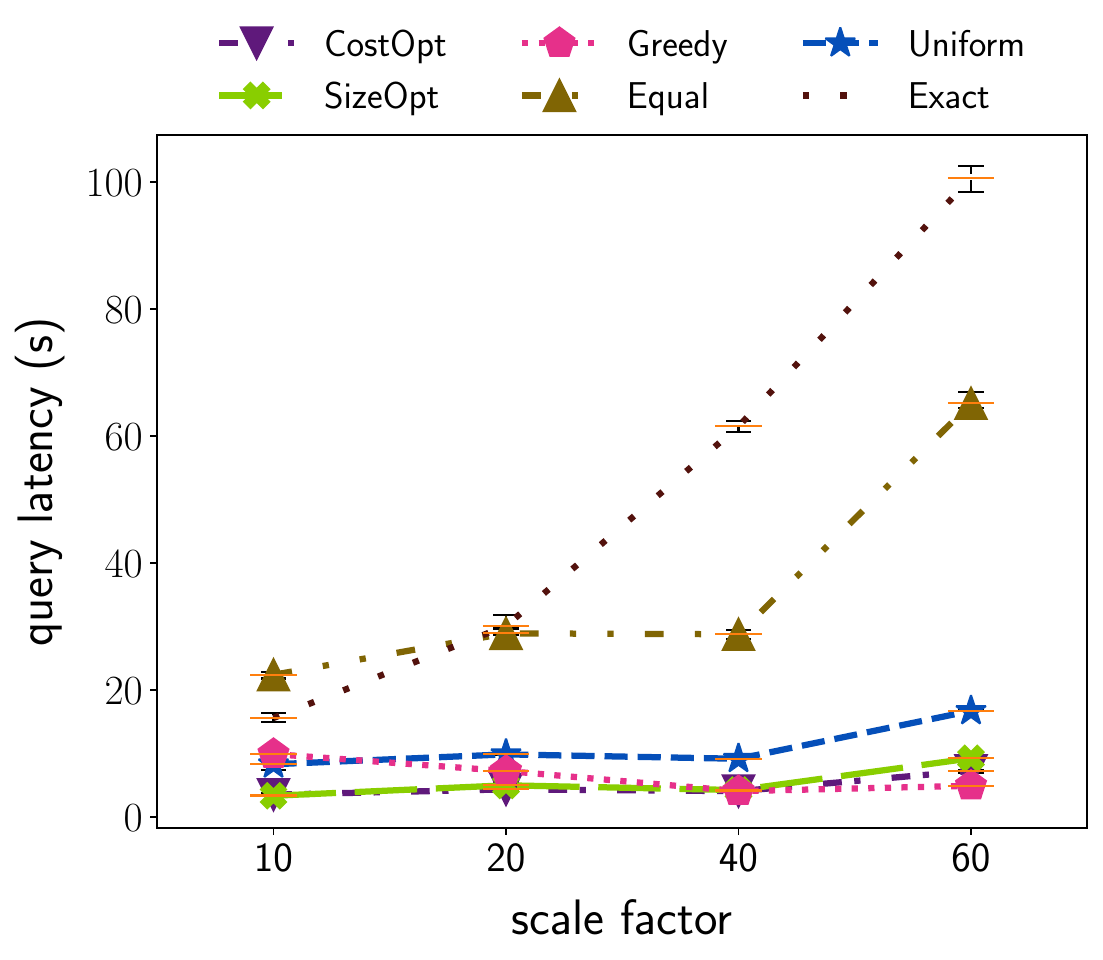}
    \end{minipage} \\
	\vspace{-2.8 mm}
    \subfigure[Varying Scale Factors]{
        \includegraphics[width=.47\linewidth]{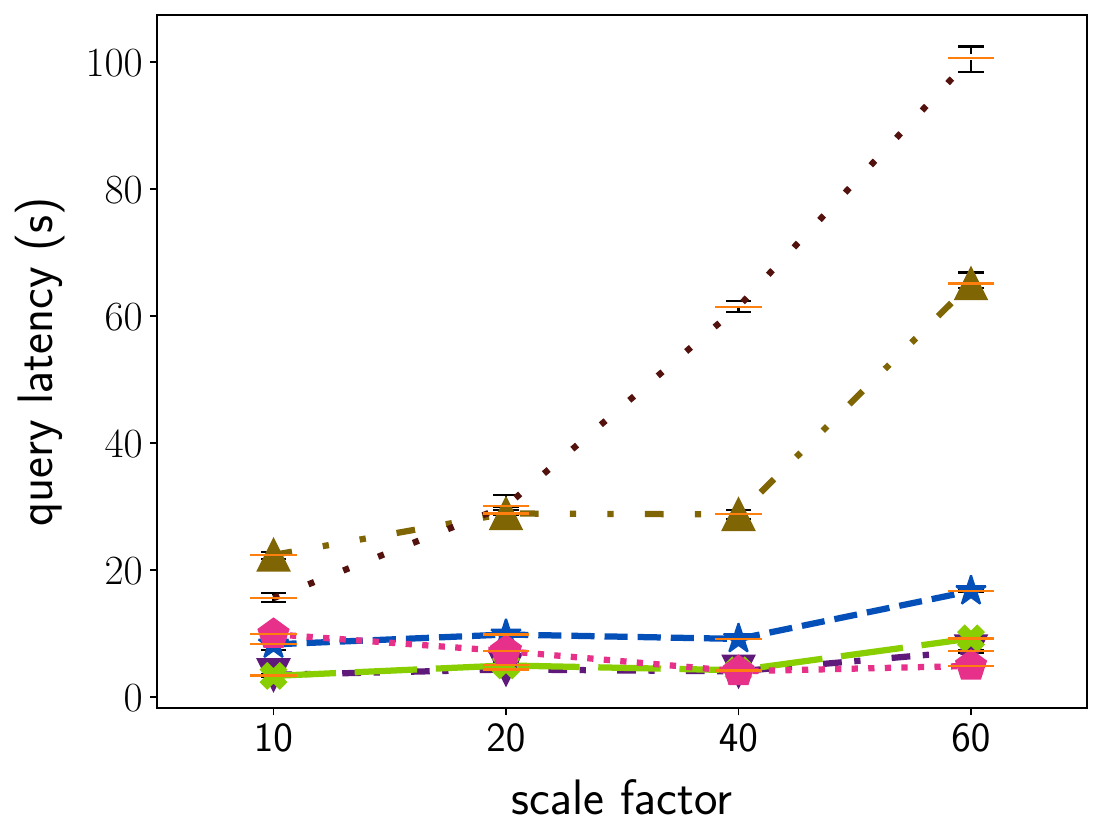}
    \label{fig:scalability} 
    } \hfill
    \subfigure[Varying high variance regions]{
        \includegraphics[width=.455\linewidth]{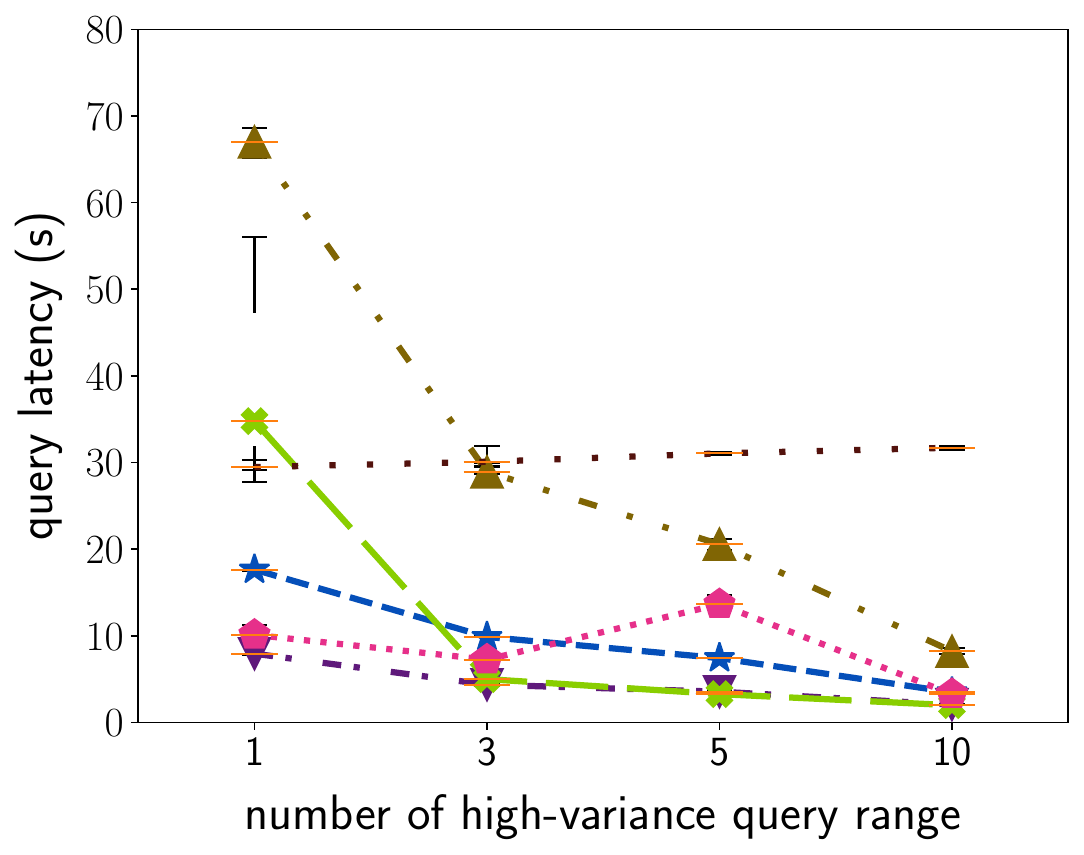}
        \label{fig:tpch_overall_variance}
    }
    \vspace{-15pt}
	\caption{\revc{TPC-H latency}}
    \vspace{-8pt}
\end{figure}

\begin{figure*}[t]
    \centering
    \subfigure[flight]{
        \includegraphics[width=.28\linewidth]{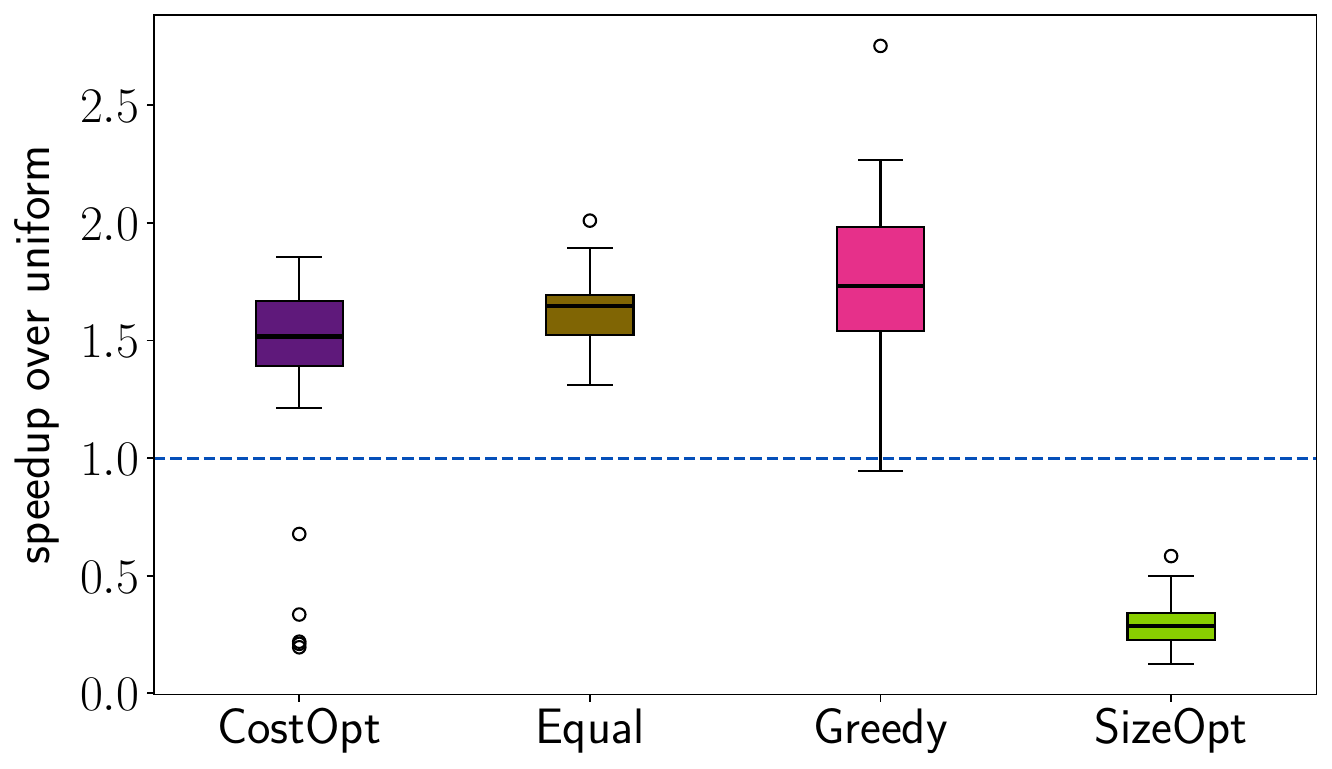}
        \label{fig:speedup_flight}
    }\hfill
        \subfigure[census]{
        \centering
        \includegraphics[width=.28\linewidth]{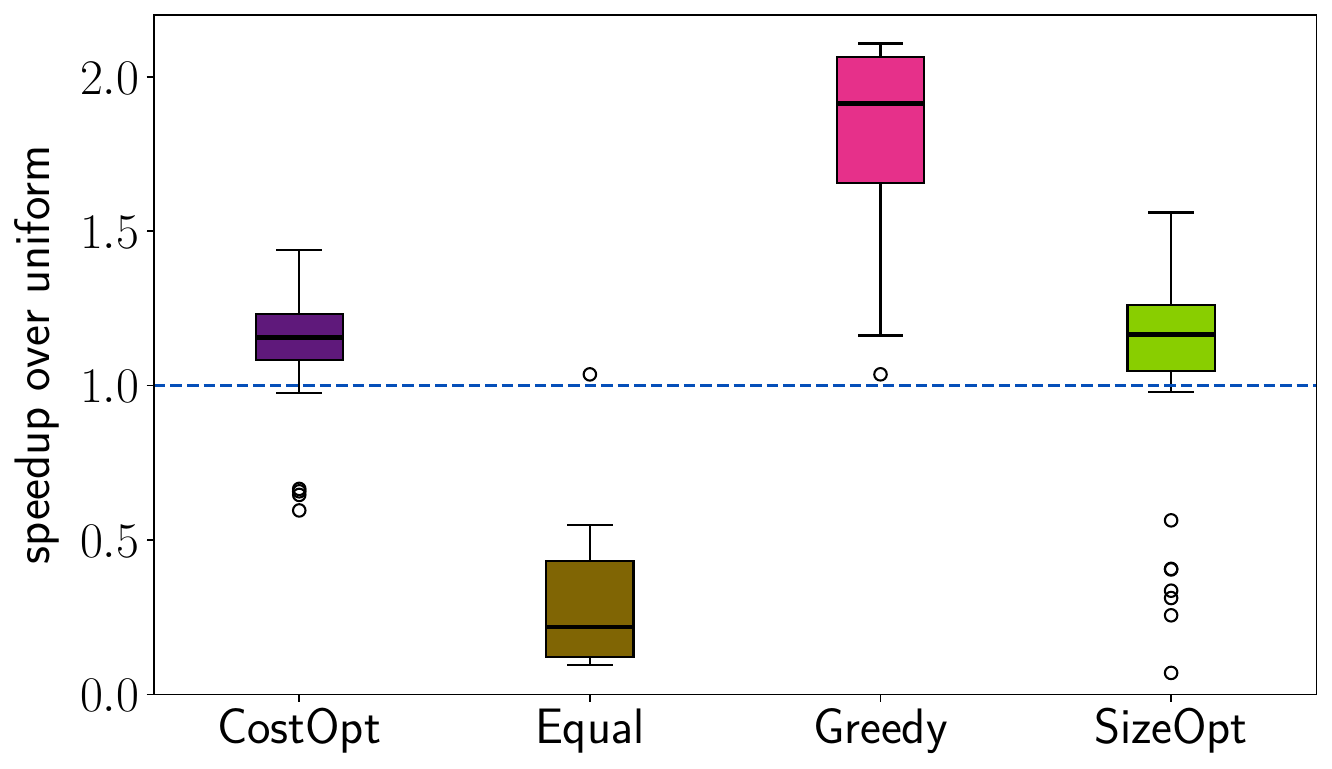}
        \label{fig:speedup_census}
    }\hfill
    \subfigure[TPC-H]{    
	\centering
        \includegraphics[width=.28\linewidth]{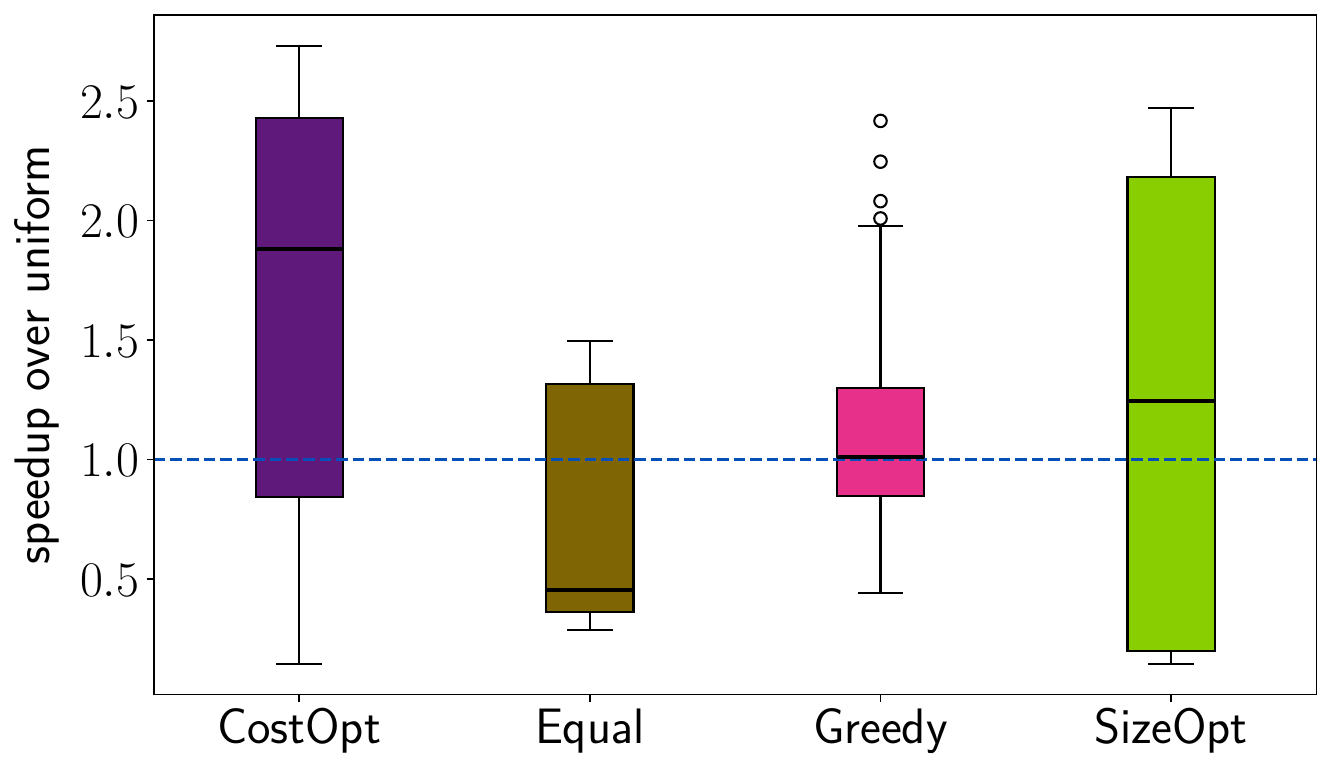}
         \vspace{-15 pt}
	\label{fig:speedup_tpch}
    }
	\vspace{-15 pt}
    \caption{Speedup over uniform with randomly generated query ranges}
    \vspace{-15pt}
    \label{fig:speedup}
\end{figure*}

We also vary the number of high shipping delay date ranges we generate
in the lineitem table of TPC-H (SF 20). The more high
shipping delay date ranges we generate, the smaller the overall
variance is because a larger proportion of data will have similar
higher values. Therefore, this experiment evaluates the margin for
improvement under different overall estimation variances. The results
are shown in Figure~\ref{fig:tpch_overall_variance}.  \texttt{CostOpt}
is consistently the best method among all. In contrast,
\texttt{Greedy} and \texttt{SizeOpt} demonstrate less robustness.
\texttt{Exact} slightly increases in time as the overall
variance increases, due to an increase of the total number of records.




\vspace{-8 pt}
\subsection{\revb{Speedup on randomly generated queries}}

\revb{
In previous experiments, we mostly show our speedup for queries with high
estimator variances to begin with, as our method is designed for that. However,
to further demonstrate the overall performance improvement in realistic workload,
we further show its speedup over a set of randomly generated queries over
these datasets.
Specifically, Figure~\ref{fig:speedup} shows the speedup of
\texttt{CostOpt}, \texttt{Equal}, \texttt{Greedy}, \texttt{SizeOpt} and
\texttt{Uniform} on flight, census and TPC-H lineitem table while randomly generated query
ranges, where the speedup is computed as the ratio of the query latency of other
methods to the query latency of \texttt{Uniform}. \texttt{CostOpt} is the most
robust method while \texttt{Greedy} is mostly robust with a similar speedup. As expected,
some are less optimizable due to small estimator variance, so our methods could be less effective or even cause a small slow down. However, in most query ranges on these real-world
and synthetic datasets, the estimator variance is high enough for these methods
to have sizable speedup. We also want to note that, in rare cases, there are
outliers of these methods such that the slow down is significant. However, they
are easy to mitigate in practice: As we can compute whether the resulting
confidence interval is close to what we expect for a particular stratification 
strategy during Phase 1, we can run Phase 1 in incremental steps similar to traditional online aggregation~\cite{hellerstein97:_onlin},
compute the resulting confidence interval and compare it against the estimated
confidence interval based on Phase 0 statistics. If they differ too much, then
we can immediately switch back to \texttt{Uniform} to avoid further slow down.
In contrast, the two baseline methods \texttt{Equal} and \texttt{SizeOpt}
are significantly volatile in speedup and can have larger proportion that
causes slow down, because they do not have optimality guarantees.}

\vspace{-8 pt}
\subsection{Parameter tuning}
\label{sec:hyperparameter_eval}

\begin{figure}[t]
    \subfigure[Query latency]{ \centering
    \includegraphics[width=.48\linewidth]{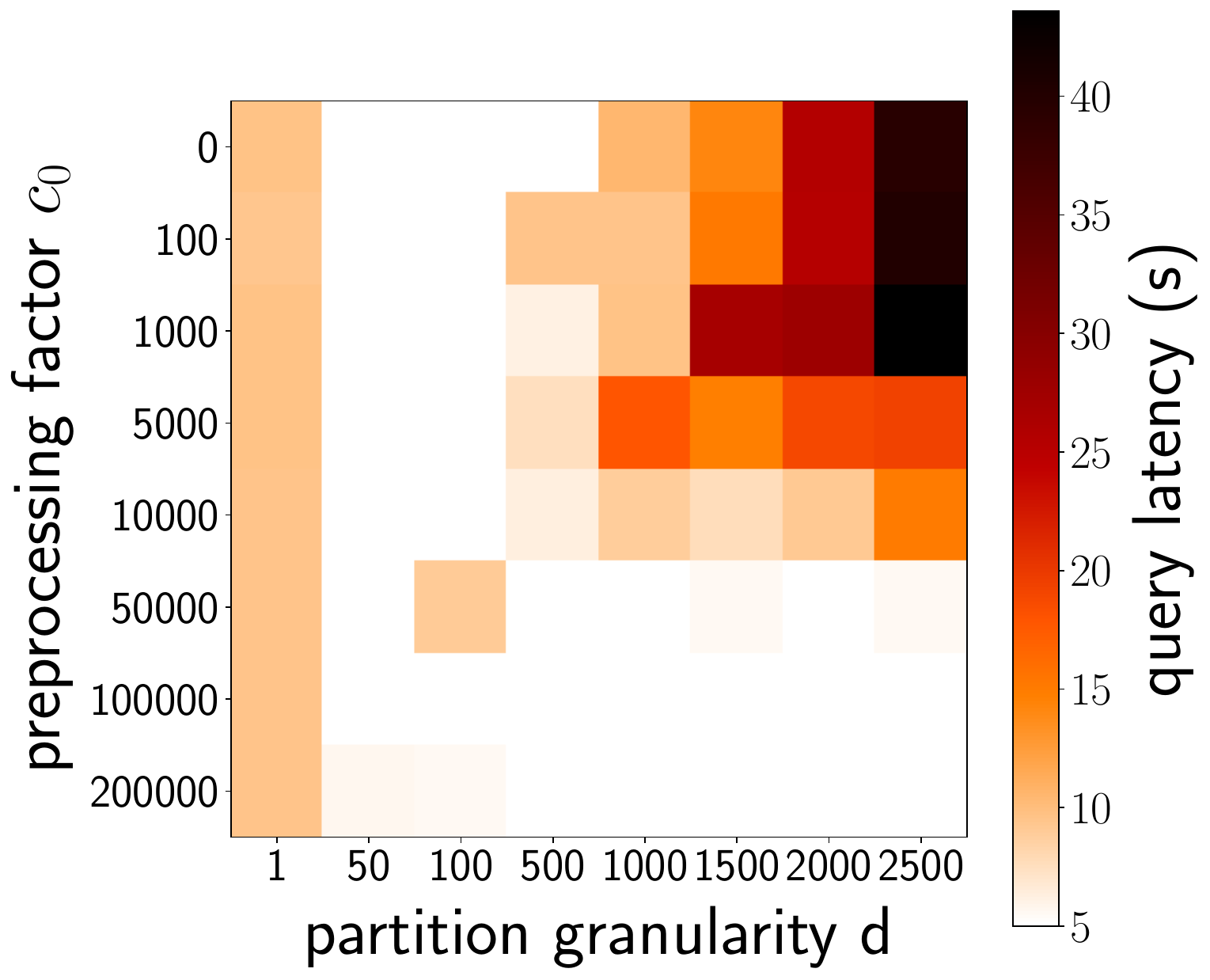}
    \label{fig:total_time_heapmap} }\hfill
    \subfigure[Total sample size]{ \centering
    \includegraphics[width=.48\linewidth]{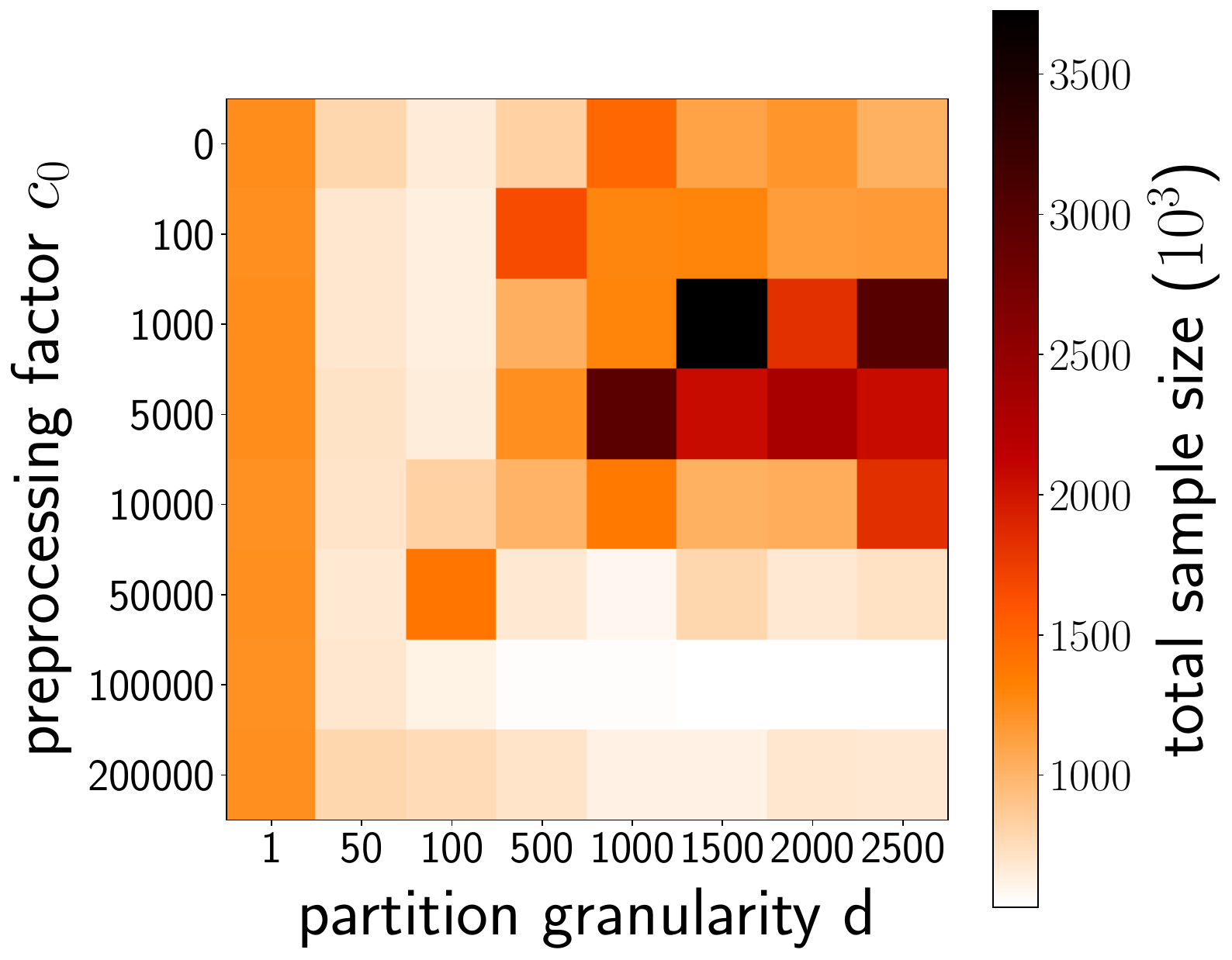}
    \label{fig:sample_size_heatmap}}\\
    \vspace{-10pt}
    \subfigure[Optimization time]{\centering
    \includegraphics[width=.48\linewidth]{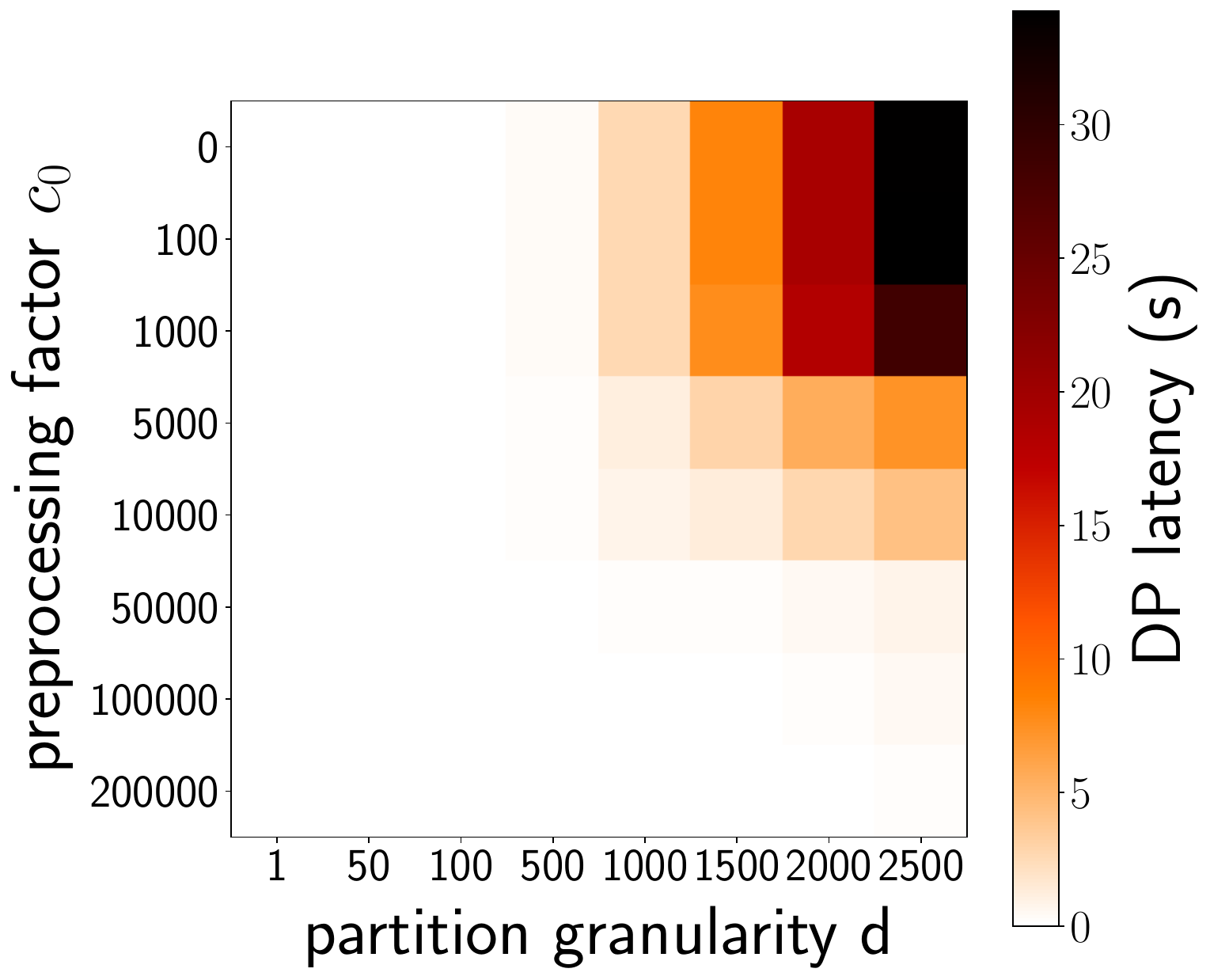}
    \label{fig:dp_time_heatmap} }
    \subfigure[Resulting number of partitions]{ \centering
    \includegraphics[width=.48\linewidth]{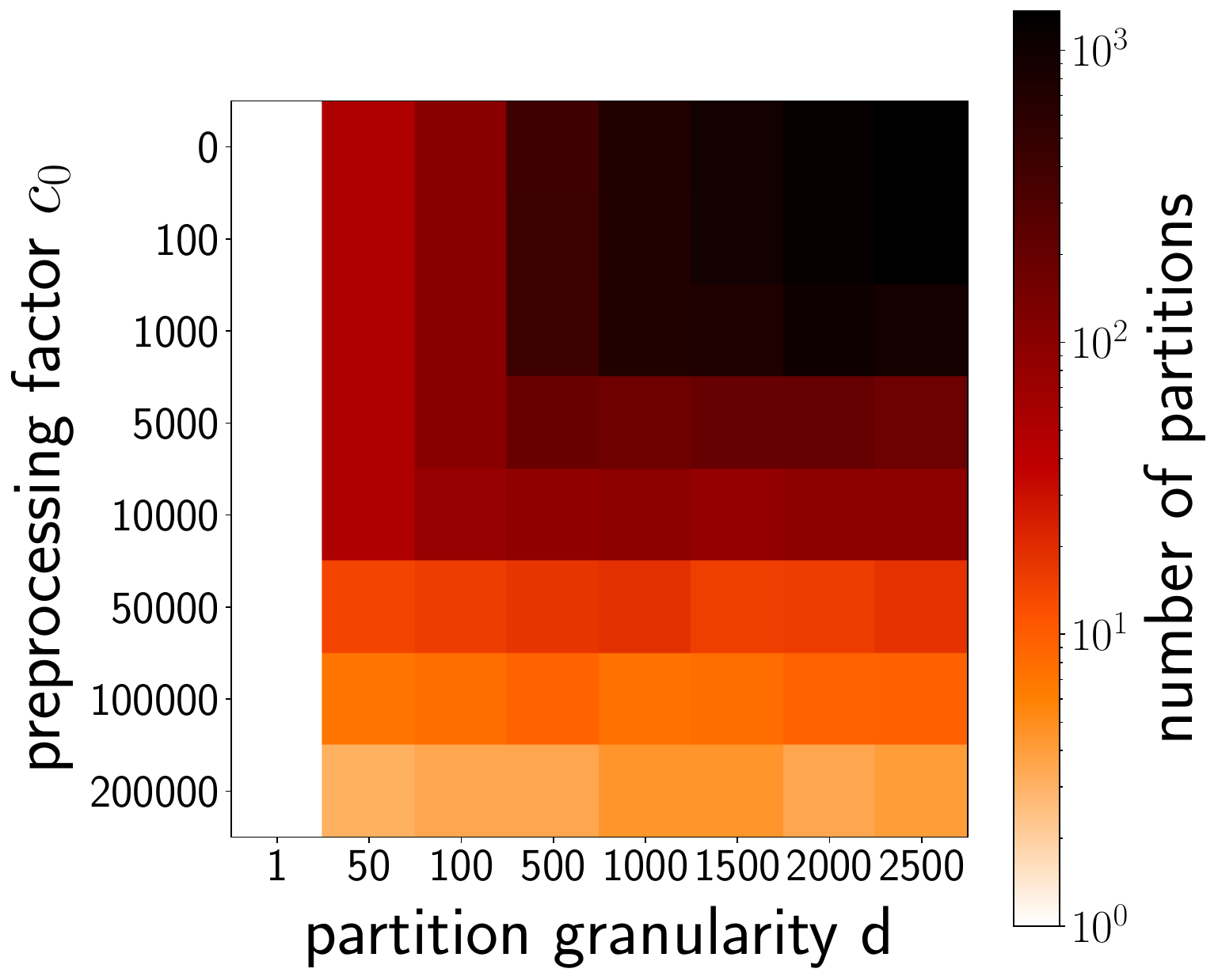}
\label{fig:num_partitions_heatmap} }
    \vspace{-15pt}
    \caption{Varying parameters for \texttt{CostOpt}} 
    \label{fig:hyperparameter_tuning}
\end{figure}

\begin{figure}[t]
    \begin{minipage}{.48\linewidth}
    \centering
    \includegraphics[width=\linewidth]{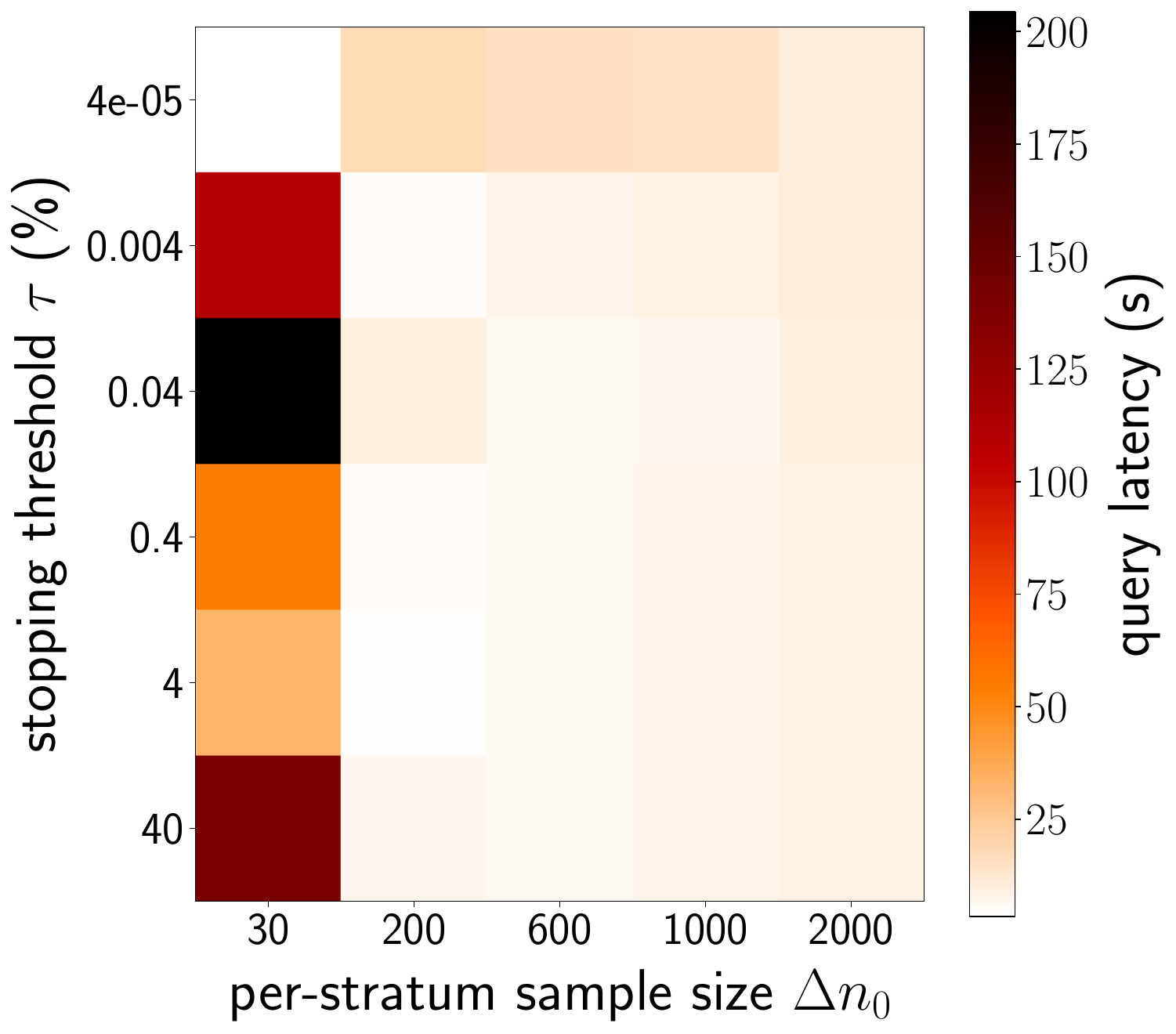}
    \vspace{-25 pt}
    \caption{Varying parameters for \texttt{Greedy}}
    \label{fig:hyper_tree}
    \end{minipage}\hfill
    \begin{minipage}{.48\linewidth}
    \centering
    \includegraphics[width=\linewidth]{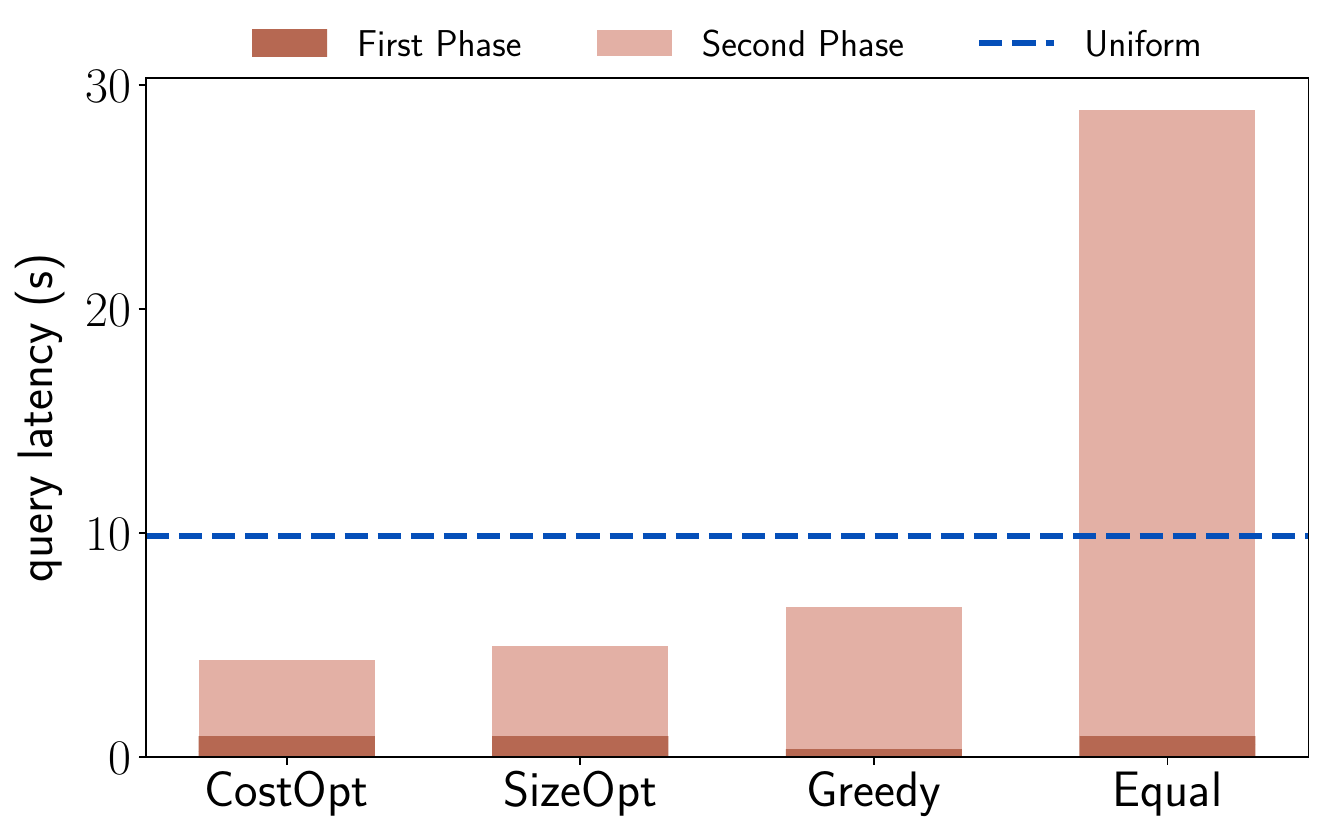}
    \vspace{-10 pt}
    \caption{Execution time breakdown for different stratification strategies over TPC-H} 
    \label{fig:breakdown}
    \end{minipage}

\end{figure}

We next investigate the impact of algorithm parameters
on \texttt{CostOpt} and \texttt{Greedy} using TPC-H lineitem scale factor 20
with 3 high-variance query ranges.

\revc{\Paragraph{Hyper-parameter selection.} The selection of hyper-parameters
is guided by the trade-offs between optimization overhead and benefit. As
demonstrated in Figure~\ref{fig:hyperparameter_tuning} and
Figure~\ref{fig:hyper_tree}, our methods remain robust performance across a
large amount of hyper-parameters. The default settings are recommended because
it consistently provides good trade-offs and low query latencies. For
\texttt{CostOpt}, preprocessing factor $c_0$ represents the constant overhead
of tree traversals and data structure settings. When migrating to a new
environment, it requires only a one-time calibration, which is similar to other
cost constants in optimizer. Partition granularity $d$ directly influences the
$O(d^3)$ dynamic programming algorithm. We find that while increasing $d$ can
do produce slightly better estimation, it makes optimization too expensive with limited
return, which leads to increasing in query latency. So, we recommend using our default
settings for $d$. However, if a user finds that a particular $d$ does not work
too well, then they may want to maximize $d$ until optimization
overhead becomes unacceptable (as projected based on the $O(d^3)$ complexity).
In contrast, \texttt{Greedy} are more sensitive to its hyper-parameters because
the decrease of estimator variance in each stratification step highly depends
on the data distribution.  In real deployment, we recommend using a reference
workload or historical workload to identify the hyper-parameter that works best
overall.}

\reva{\Paragraph{Query latency breakdown.} Figure~\ref{fig:breakdown} provides
a breakdown of query latency under our default settings for different methods
on TPC-H query.  First phase represents the time for drawing samples in Phase 0
and the subsequent stratification optimization time. Second phase is the
remaining time for fetching samples in Phase 1 until the approximate
aggregation reaches the specified relative CI 0.01.  For reference, we also show
the total execution time with non-stratified method \texttt{Uniform} as a
dashed line. In general, if we set parameters of each of the methods to have less than $1$
second Phase 0 time, \texttt{CostOpt} and \texttt{Greedy} typically have good
speedup compared to \texttt{Uniform}. The improvement of \texttt{SizeOpt} and
\texttt{Equal} can vary and highly depend on datasets and query.
}

\begin{figure}[t]	
        \begin{minipage}{\linewidth}
		\centering
		\includegraphics[width=0.7\linewidth, clip = true, trim=50 395 50 10]{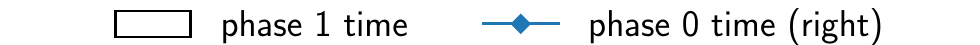} 
	\end{minipage}\\
	\subfigure[flight]{ \centering
	\includegraphics[width=.495\linewidth]{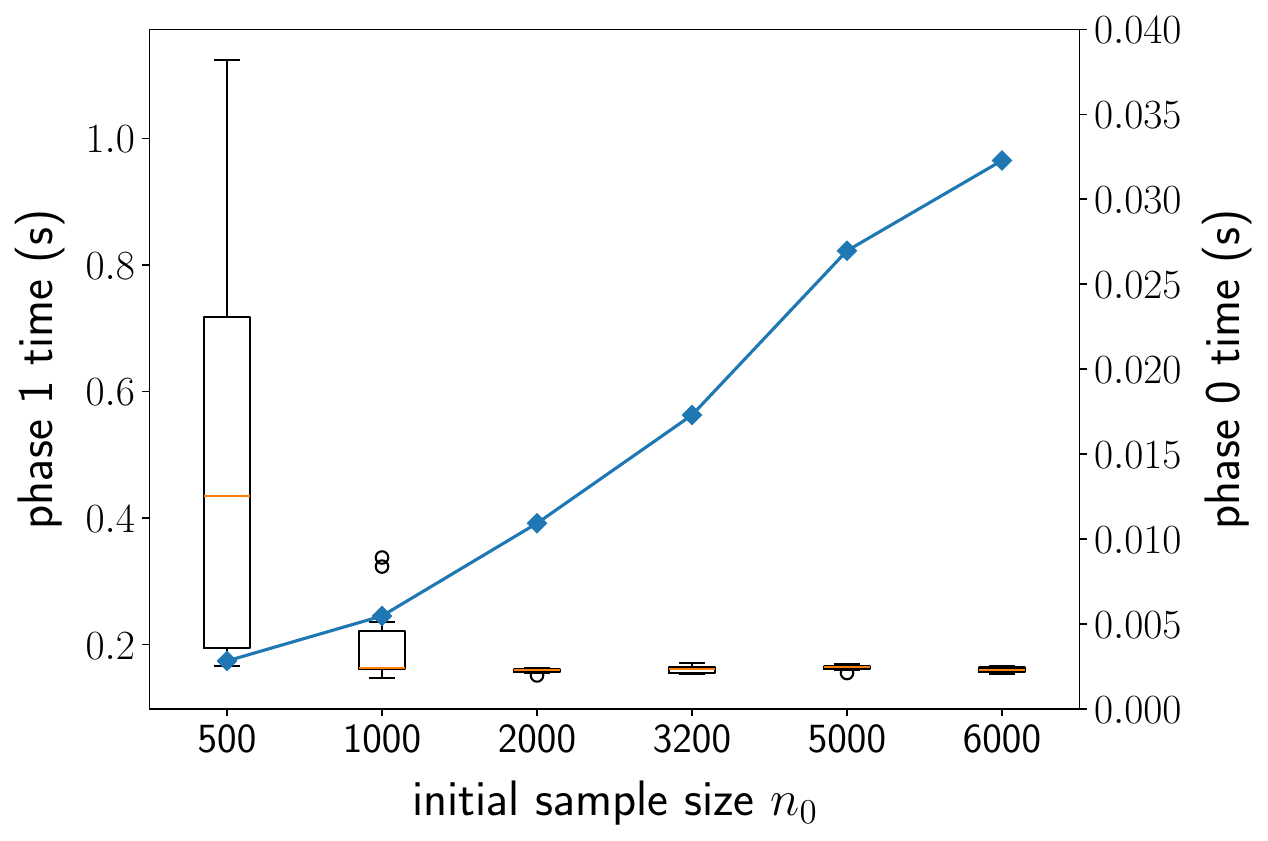}
    \label{fig:n0_flight} }\hfill \subfigure[TPC-H]{ \centering
    \includegraphics[width=.465\linewidth]{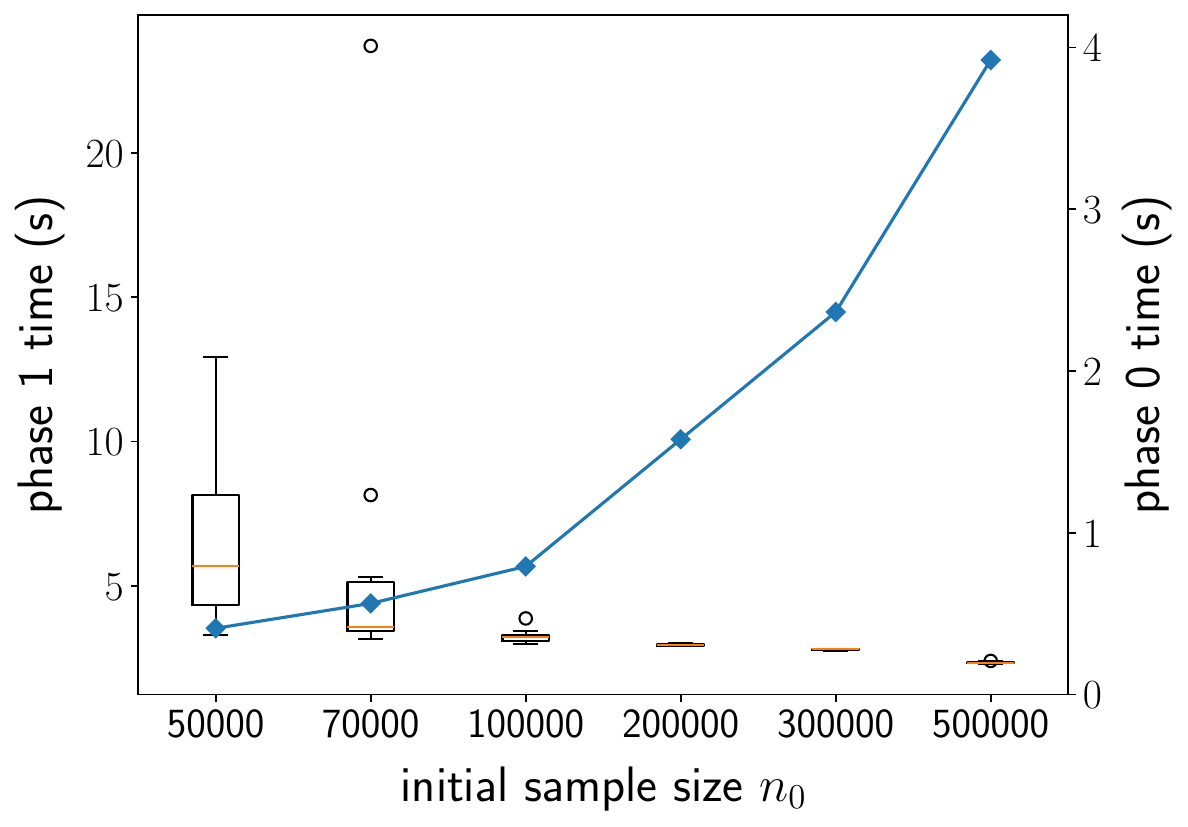}
    \label{fig:n0_tpch}}\hfill
    \vspace{-10 pt}
    \caption{Varying the initial sample size $n_0$ for \texttt{CostOpt}}
        \vspace{-3pt}
    \label{fig:varying_n0}
\end{figure}



\reva{\Paragraph{Impact of varying initial sample size $n_0$.} Clearly, the initial
sample size can greatly impact the distinct values we find and the accuracy
of statistics we estimate for stratification optimization, which in turn can
impact the effectiveness and the overall speedup. We evaluate the impact of initial sample size $n_0$ in \sysname on the
flight dataset and the lineitem table (SF 20) with 3 high
delay ranges.
For each
different $n_0$, we run the experiments for 10 times and plot the error bars in
Figure~\ref{fig:varying_n0}. The actual time for
first phase includes optimization time. As shown in the figures, as $n_0$
increases, the actual time for the second phase becomes more stable. However,
if too many $n_0$ is requested, the actual time for the first time will
increase without reducing much of the actual time of the second phase.
In practice, if distinct values are not stored separately,
we recommend setting $n_0$ to $\tilde{O}(NDV)$ such that every distinct
value will have a sufficient support for accurate statistics estimation in Phase 0.
However, this may become too expensive over datasets with too many distinct values,
so practitioner may also want to bound the Phase 0 time with a constant cost.
Hence, in our experiments, we choose the smaller of $200 NDV$, or $100,000$ (which
roughly corresponds to $1$ second of Phase 0 time on average), but these can be
further adjusted based on actual need and dataset. In addition, if the distinct
values already exist separately, e.g., encoded in a dictionary or catalog, then
our Phase 0 execution can be modified to directly draw $200$ samples for
each distinct value.
}

\vspace{-8 pt}
\section{Related work and possible extension}
\label{sec:related}

There are many works that optimize approximate query latencies in
different AQP system settings. Most related are those that design
stratified sampling strategies.  Previous works materialize offline
samples used repeatedly for approximate
queries~\cite{10.1145/3183713.3196905,DBLP:conf/eurosys/AgarwalMPMMS13}
or draw online samples with table
scans~\cite{Kandula:2016:QLA:2882903.2882940} , where the
stratification is done at the finest granularity with each distinct
value on a set of stratification columns. However, as we show in the
experiments, such strategy can cause too much per-partition overhead
on queries with too many distinct values. Note that these works do not
utilize the optimal Neyman allocation.

Closest to our work is due to Surajit et
al.~\cite{10.1145/1242524.1242526} but it was designed for
materializing an offline sample given a
query workload. In addition, their optimization target is also
different than ours because our setting, in which we optimize for
ad-hoc queries with fast index-assisted scans, significantly differs
from traditional AQP system setting that scans. Their approach is to first
create the fundamental regions given a workload (i.e., regions that
are read entirely or not at all by all queries in the workload), and
then further divide it up into finer regions to $6$ levels. This
approach is not feasible with index-assisted sampling because the
sampling index used in index-assisted sampling restricts what
strata (i.e., must be ranges for IRS index) we can efficiently
sample from. \cite{10.14778/3625054.3625059} also performs sample size
allocation but they focus on stratification on group-by and join
columns and uses an auto-encoder to generate synthetic samples with
similar data characteristics instead of drawing samples from the
database, which will highly depends on the quality of trained
autoencoder.

\cite{10.1145/3448016.3457277,10.1145/3183713.3183747} takes a different route
to selectively materializes subrange query results for a query workload, and
combines online scan-based sampling with pre-aggregated results, which are
orthogonal to this work and we may explore in the future.  Machine learning
based AQP
systems~\cite{10.14778/3384345.3384349,10.14778/3421424.3421432,10.1145/3299869.3324958}
and deterministic AQP systems~\cite{DBLP:conf/cidr/JoT20,DBLP:conf/sigmod/JoT20}
utilize pre-built models or data encoding to perform to perform queries over a
specific query template, which can often lead to extremely low query latencies
and high accuracy on the data and query trained or designed for. However, they
may require specialized training or customization for workloads.
PilotDB~\cite{10.1145/3725335} also performs a similar two-phase sampling to
first determine the estimator variance to decide how many samples to fetch in
the second phase. However, PilotDB uses \texttt{SYSTEM} sampling which in
general increases sample variance by sampling large blocks one at a time. More
importantly, it does not perform stratification to optimize latency for
high-variance data as \sysname does.

Our current work builds efficient IRS index
designs~\cite{olken93:_random,hu14:_indep,afshani_et_al:LIPIcs:2017:7859,afshani_et_al:LIPIcs:2019:10408,10.1145/3448016.3452806,10.1145/3626744,10.14778/3617838.3617840}.
We specifically build
upon~\cite{10.14778/3538598.3538606,10.14778/3611540.3611602} which enables
concurrent approximate queries over frequent updates.

\Paragraph{Possible extensions.}
While this work focuses on optimizing ad-hoc single-table approximate
aggregation queries where the table can be sampled with an IRS index,
many different extensions are possible by combining the more general
IRS and IQS (Independent Query Sampling, e.g., join sampling) indexing
techniques with our design. 
However, all such extensions may have
non-trivial design trade-offs beyond the scope of this work, and we
would like to further explore these directions in the future.

For example, Then this work can be naturally extended to a join
aggregation query with high estimator variance.  However, it is also
possible that more fine-grained stratification based on the (subset of)
join attributes combined with a predicate on the starting table may
help reduce high variance induced by join skewness. This work can also be extended to support
more general range queries such as spatial range aggregation queries
leveraging SIRS (Spatial Independent Range
Sampling)~\cite{10.1145/3448016.3452806} indexes.  However, in higher
dimensions, how to compose larger strata from smaller strata may become a more
complex optimization decision, which needs to be investigated further.  This
work may also extend to group-by queries with two possible strategies: (1) We
stratify on a composite key of group-by and range columns on which there is an
IRS index, and perform per-group two-level index-assisted approximate query
evaluation; or (2) We perform rejection sampling for each group by only
sampling from the range column on which there is an IRS index, similar
to~\cite{wanderjoin}.  They are different trade-offs in performance and
optimization overhead that require in-depth investigation.

\vspace{-8 pt}
\section{Conclusion}
\label{sec:conclusion}

In this work, we present \sysname, an index-assisted sampling system which
improves approximate range aggregation query latencies using a two-phase
index-assisted approximate query evaluation framework with optimized stratified
sampling. Experiments show \sysname significantly improves the approximate
query cost for query ranges with high variances compared to baselines. We also
discuss multiple possible extensions to handle additional query types in future
work, including joins, group-by and multi-dimensional range predicates.

\begin{acks}
    This work is supported by NSF IIS-2339596.
\end{acks}

\bibliographystyle{ACM-Reference-Format}
\bibliography{main}

\end{document}